\documentclass[pdflatex,sn-mathphys-num]{sn-jnl}
\usepackage{mathrsfs}

\usepackage{graphicx}%
\usepackage{multirow}%
\usepackage{amsmath,amssymb,amsfonts}%
\usepackage{amsthm}%
\usepackage{mathrsfs}%
\usepackage[title]{appendix}%
\usepackage{xcolor}%
\usepackage{textcomp}%
\usepackage{manyfoot}%
\usepackage{booktabs}%
\usepackage{algorithm}%
\usepackage{algorithmicx}%
\usepackage{algpseudocode}%
\usepackage{listings}%

\usepackage{balance} 

\usepackage{algorithm}
\usepackage{algpseudocode}

\usepackage{pmat} 

\usepackage{tikz}

\usepackage{breqn}

\newcommand{\STATE}{S}
\renewcommand{\state}{s}

\newcommand{\stochastic}{\{\STATE_{\round}\}_{\round \geq 0}}

\newcommand{\round}{t}

\newcommand{\Colour}{X}

\newcommand{\textdef}[1]{\textbf{#1}}

\newcommand{\citenp}[1]{\citeauthor{#1} \citeyear{#1}}
\newcommand{\citenptext}[1]{\citeauthor{#1} (\citeyear{#1})}

\newtheorem*{remark}{Remark}
\newtheorem{corollary}{Corollary}
\newtheorem{lemma}{Lemma}

\theoremstyle{thmstyleone}%
\newtheorem{theorem}{Theorem}
\newtheorem{proposition}[theorem]{Proposition}%

\theoremstyle{thmstyletwo}%
\newtheorem{example}{Example}%

\theoremstyle{thmstylethree}%
\newtheorem{definition}{Definition}%

\begin{document}

\title[Voter Model Meets Rumour Spreading: an FPRAS for Consensus Probabilities on Voter Models with Agnostic Nodes]{Voter Model Meets Rumour Spreading: an FPRAS for Consensus Probabilities on Voter Models with Agnostic Nodes}


\author*[1]{\fnm{Marcelo} \sur{Matheus Gauy}}\email{marcelo.gauy@unesp.br}

\author[2]{\fnm{Anna} \sur{Abramishvili}}\email{anna.abramishvili@kcl.ac.uk}

\author[3]{\fnm{Eduardo} \sur{Colli}}\email{colli@ime.usp.br}

\author[4]{\fnm{Nicolaus} \sur{Heuer}}\email{nicolaus.heuer.maths@gmail.com}

\author[5]{\fnm{Tiago} \sur{Madeira}}\email{tmadeira@alumni.usp.br}

\author[2]{\fnm{Frederik} \sur{Mallmann-Trenn}}\email{frederik.mallmann-trenn@kcl.ac.uk}

\author[6]{\fnm{Vinícius} \sur{Franco Vasconcelos}}\email{vfvmat@ime.usp.br}

\author*[2]{\fnm{David} \sur{Kohan Marzagão}}\email{david.kohan@kcl.ac.uk}

\affil*[1]{\orgdiv{Departamento de Ciências e Tecnologia}, \orgname{São Paulo State University}, \orgaddress{\street{Geraldo Alckmin, 519}, \city{Itapeva}, \postcode{18409-010}, \country{São Paulo, Brazil}}}

\affil[2]{\orgdiv{Informatics}, \orgname{King's College London}, \orgaddress{\street{30 Aldwych}, \city{London}, \postcode{WC2B 4BG},  \country{United Kingdom}}}

\affil[3]{\orgdiv{Departmento de Matemática Aplicada}, \orgname{University of São Paulo}, \orgaddress{\street{Rua do Matão, 1010}, \city{São Paulo}, \postcode{05508-090}, \country{São Paulo, Brazil}}}

\affil[4]{\orgdiv{Independent Researcher}, \orgaddress{\city{London}, \country{United Kingdom}}}

\affil[5]{\orgdiv{Departmento de Ciência da Computação}, \orgname{University of São Paulo}, \orgaddress{\street{Rua do Matão, 1010}, \city{São Paulo}, \postcode{05508-090}, \country{São Paulo, Brazil}}}

\affil[6]{\orgdiv{Departmento de Matemática}, \orgname{University of São Paulo}, \orgaddress{\street{Rua do Matão, 1010}, \city{São Paulo}, \postcode{05508-090}, \country{São Paulo, Brazil}}}

\abstract{

Problems of consensus in multi-agent systems are often viewed as a series of independent, simultaneous local decisions made between a limited set of options, all aimed at reaching a global agreement. Key challenges in these protocols include estimating the likelihood of various outcomes and finding bounds for how long it may take to achieve consensus, if it occurs at all. 

To date, little attention has been given to the case where some agents have no initial opinion. In this paper, we introduce a variant of the consensus problem which includes what we call `agnostic' nodes and frame it as a combination of two known and well-studied processes: voter model and rumour spreading. We show (1) a martingale that describes the probability of consensus for a given colour, (2) bounds on the number of steps for the process to end using results from rumour spreading and voter models, (3) closed formulas for the probability of consensus in a few special cases, along with a polynomial-time algorithm for the case where the number of agnostic vertices is at most logarithmic and (4) that the computational complexity of estimating the probability with a Markov chain Monte Carlo process is $O(n^2 \log n)$ for general graphs and $O(n\log n)$ for Erd\H{o}s-Rényi graphs, resulting in a fully polynomial-time randomized approximation scheme (FPRAS) for estimating the probabilities of consensus. Furthermore, we present experimental results suggesting that the number of runs needed for a given standard error decreases when the number of nodes increases.
}


         
\newcommand{\BibTeX}{\rm B\kern-.05em{\sc i\kern-.025em b}\kern-.08em\TeX}

\maketitle 


\section{Introduction}

In multi-agent consensus problems, agents make a sequence of independent and autonomous choices from a finite set based on their local information. Agents have a shared goal of reaching a consensus state, in which they all represent the same choice. The process is often abstracted as a graph in which nodes represent agents, their colour represent their current choice (or opinion, or state), and edges of the graph represent visibility or influence between agents.

Multi-agent processes on graphs have been shown to have several applications, 
including autonomous robots or drones~\cite{yan2013survey, ismail2018survey},  electrical flow estimation \cite{becchetti2018pooling, cruciani2021phase},
mutation fixation in biology \cite{moran1958random, lieberman2005evolutionary}, among others. 
In the simplest of such processes, the \textit{voter model}, at each point in time (called `round') nodes may change their colour based on the opinion of their neighbours until consensus is reached (e.g. the case where all nodes share the same colour). More formally, this process can be either synchronous or asynchronous. In the asynchronous case, a node is chosen uniformly at random and selects a neighbour proportional to the weight of the edge between then. It then adopts the colour (or opinion) of the chosen neighbour. In the synchronous case, all nodes act simultaneously and independently (i.e., the choice of one does not affect the choice of the other in the same round). Consensus processes in multi-agent systems have been extensively studied (e.g. \cite{martinez2005synchronous,lynch1996distributed, cao2015event,olfati2007consensus, marzagao2017multi}).

Given an initial colour configuration, the probability of consensus and time-bounds for the number of rounds before such consensus is achieved for a given colour are some of the core problems studied in this domain. Extensive results have been found for both synchronous \cite{hassin2001distributed} and asynchronous \cite{cooper2016linear} process, as well as for processes with undecided states (\cite{angluin,perron, clementi_et_al:LIPIcs.MFCS.2018.28, petra1}). In these, nodes do not change directly from one colour to another but transition via an `undecided' state in between them.  

One of the features of the classical voter model is that it assumes all nodes start off with a colour/opinion. In some scenarios, we may want to model a process in which, at the start, some nodes do not have an opinion at all, which may be different from being `undecided' after been given a set of option, as they have not been in contact with any of the opinions in this process. One can think of examples related to election scenarios in which voters do not yet know the candidates and may be therefore influenced by the first contact with a candidate. Or in a process on a blockchain, in which new blocks are mined and the information of new mined blocks traverse the network, possibly competing with other new blocks mined at a similar time. 

In this paper, we introduce and study a variant of the voter model in which nodes can have an extra state, which we call `agnostic' (represented by the colour, say, white). We call it \textit{voter model with agnostic nodes}. Agnostic nodes become gnostic if they choose a gnostic neighbour to copy its colour. Once gnostic, they can never become agnostic again, i.e., if a gnostic node chooses an agnostic one, the gnostic node keeps its current colour and nothing happens. For a precise definition of the problem, see Section~\ref{subsec:agnostic}. For an online interactive visualisation of the process, see \url{https://connections.computer/demo.html}.  The main challenge posed when studying this variant is that there is an asymmetry between states, in that agnostic states can become gnostic but not vice versa. We provide an efficient Monte Carlo algorithm for estimating the probability of consensus in this voter model with agnostic states for any graph and any initial configuration. Specifically, we show that this Monte Carlo algorithm provides a fully polynomial-time randomized approximation scheme for the problem of computing the probability of consensus in the voter model with agnostic states, under very mild assumptions of the transition probabilities. Furthermore, we provide time-bounds for consensus to be achieved. Lastly, we present a polynomial time algorithm for the case where the initial number of agnostic nodes is at most logarithmic in the number of nodes in the graph.

The variant we study can also be seen as a generalisation of the rumour spreading model~\cite{feige1990randomized, fountoulakis2010reliable, acan2015push, panagiotou2017asynchronous}. In it, there are nodes that are `informed' and nodes that are `uninformed' and the process studies the time bounds until all nodes become informed. Like in our variant, an 'informed' node cannot become 'uninformed'. Our voter model with agnostic nodes is analogous to two or more rumours that compete not only to gather more agnostic nodes but also to flip the opinion of other gnostic nodes. Alternatively, one can also see the voter model with agnostic nodes as a combination of two models happening simultaneously: the classical voter model with a rumour spreading process. 

The following motivates the problem with a toy example.

\tikzset{red/.style={draw=red, circle, fill=red!30, inner sep=1}}
\tikzset{blue/.style={draw=blue, circle, fill=blue!30, inner sep=1}}
\tikzset{white/.style={draw=black, circle, fill=white, inner sep=1}}
\tikzset{green/.style={draw=green, circle, fill=green!30, inner sep=1}}
\tikzset{orange/.style={draw=orange, circle, fill=orange!30, inner sep=1}}
\tikzset{yellow/.style={draw=yellow, circle, fill=yellow!30, inner sep=1}}

\begin{example}\label{exm:motivation}
    Consider the graph and initial configuration depicted in Figure \ref{fig:motivational_example}. In this example, each node chooses a neighbour with uniform probability. For example, $v_2$ has $50\%$ chance of choosing $v_1$ and thus becoming blue at round $S_1$ and $50\%$ chance of choosing $v_3$ and thus becoming red at round $S_1$. If a node chooses an agnostic node, their colour does not change. For example, if $v_1$ chooses $v_2$, $v_1$ will stay blue. 

    With that in mind, what is the probability that there will be, say, a red consensus? What can we say about the expected number of rounds until that happens?
    
    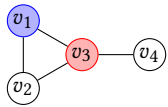
\begin{figure}[b]
    \centering
    \begin{tikzpicture}[scale = 1.5]
    
        \node[blue] (A) at (-0.52, 0.3) {$v_1$};
        \node[white] (B) at (-0.52, -0.3) {$v_2$};
        \node[red] (C) at (0, 0) {$v_3$};
        \node[white] (D) at (0.6, 0) {$v_4$};

        \draw (A) -- (B);
        \draw (A) -- (C);
        \draw (B) -- (C);
        \draw (C) -- (D);
    \end{tikzpicture}
    \caption{A motivational example of an undirected graph with an initial configuration $S_0 = s_0$ consisting of one blue node ($v_1$), one red node ($v_3$), and two agnostic nodes ($v_2$ and $v_4$). Transition probabilities are uniform, i.e., $v_3$ has $\frac{1}{3}$ chance of choosing a given neighbour, whereas $v_4$ chooses $v_3$ and becomes red with probability $1$. What are the probabilities of consensus in this case?}
    \label{fig:motivational_example}
\end{figure}
\end{example}

We will return to Example \ref{exm:motivation} later in this paper. Our main contributions of this paper are as follows:

\begin{enumerate}
    \item We present a general method for computing the consensus probabilities for any voter model with agnostic states (Theorem~\ref{thm:probability}). For the case in which the underlying (weighted) graph represents a reversible Markov chain, this method produces a martingale (Corollary~\ref{thm:martingale}). Both of these appear to be computationally intensive.
    \item Although the previous method may not be efficiently computed in the general case, we prove a closed formula for the consensus probability in the case of a complete graph with an asynchronous process (Corollary~\ref{corollary:complete_graph}).
    \item We show that the existence of agnostic nodes does not affect the complexity bounds for the expected time of achieving consensus, as the agnostic nodes typically disappear faster than consensus is achieved (Lemma~\ref{lemma:time_bounds} and Corollaries~\ref{corollary:rumour_general} and~\ref{corollary:rumour_general_matrix} and Proposition~\ref{prop:rumour_random}). This is analogous to standard known results for rumour spreading processes, and holds under very mild assumptions on transition probabilities.
    \item We present a Markov chain Monte Carlo (MCMC) algorithm (Algorithm~\ref{alg:MCMC} of Section~\ref{subsec:MCMC}) to efficiently compute the consensus probability based on the fact that agnostic nodes disappear quickly (Section~\ref{sec:time_bounds}). This algorithm is shown to provide a fully polynomial-time randomized approximation scheme (see Section~\ref{subsec:complexity}) for the consensus probability computation problem.
    \item We identify and correct a small mistake in~\cite{manohara2025viral}. The authors erroneously claim that certain probabilities are independent, which leads to slightly wrong conclusions. We present a simple fix, though we are forced to settle for a randomized approximation, instead of a deterministic one. The details are present in Section~\ref{subsec:complexity}, with Section~\ref{subsubsec:examples} presenting an example of why their independence property does not hold.
    \item We present a polynomial time algorithm for computing the consensus probability in the case where the initial number of agnostic nodes is logarithmic in the size of the underlying graph.
    \item We present an experimental analysis to support results using Algorithm~\ref{alg:MCMC}, showing that few runs are necessary to obtain good probability estimates (Section~\ref{sec:experiments}). It also suggests estimates get better as graph sizes increase.
\end{enumerate}

\section{Background and Main Definitions}

   In this section, we present concepts and results from the literature that will be used in subsequent sections. We first introduce the classical version of consensus protocol used in this paper, also known as voter model~\cite{donnelly1983finite, nakata1999probabilistic, hassin2001distributed, cooper2016linear,DBLP:journals/talg/KanadeMS23,DBLP:conf/analco/OliveiraP19}, in which all vertices have an initial opinion. Then, we propose a variant of the voter model in which some vertices do not have an initial opinion, which we label agnostic vertices. The voter model has been widely studied in the context of multi-agent systems. The winning probabilities of each colour and bounds on the convergence time were obtained for undirected graphs by Hassin \& Peleg~\cite{hassin2001distributed}. Cooper \& Rivera extended this work to the linear voting model, which captures digraphs as well as several similar consensus processes~\cite{cooper2016linear}. 
    
    \subsection{Classical Voter Model} \label{subsec:voter}
    The (pull) \textdef{voter model} defines a round-based consensus process on a strongly connected directed graph $G = (V,E)$.\footnote{Henceforth, we assume all graphs are strongly connected unless stated otherwise.} In such processes agents are represented by nodes in this graph. At each round, each node has a colour associated to it, representing the respective agent's current state (or opinion). Their goal is to reach consensus, i.e., a situation where every agent is in the same state. To that end, at each round, all agents update their state synchronously based on the colour of their out-neighbours.\footnote{For precision, we consider that agents change their state at the end of each round, after all nodes have made their decisions.}
    The probability that $v$ copies colour of node $u$ in a given round is represented by the weight of edge $(v,u)$. The weights of edges starting at a given node are assumed to be positive and to sum to $1$. We collate all these probabilities in an out-matrix $H$, which can be also seen as the adjacency matrix of $G$ where entry $H(v, u)$ represents the weight of edge $(v, u)$. We adopt the notation $H(v,u) = 0$ if $(v,u) \notin E$, and note that self loops are allowed and thus $v$ may adopt its own colour. Once reached, a consensus is stable. 

    Let $\Colour = \{c_1, \dots, c_k\}$ be the set of all possible colours on a consensus process. A \textdef{configuration} on a graph $G=(V,E)$ is a function $\state \in \Colour^V$ that associates each node $v \in V$ with a colour $c \in \Colour$, i.e., $\state(v)$ represents $v$'s colour in configuration $\state$. More formally, a process is a sequence of random variables $\stochastic$, with  $\STATE_{t+1}\in \Colour^V$ being a configuration generated based on $\STATE_t$. We say colour $i$ \textdef{wins} the process if a configuration $\STATE_t =\state$, such that $\state(v) = i$ for all $v$, is reached.
    Here, we assume processes converge with probability $1$. For discussion of fringe cases and work related to graphs in which processes may not converge, see, e.g., \cite{marzag2021influence}.

    Observe that the out-matrix $H$ of the graph $G$ can be seen as the transition matrix of a time homogeneous Markov chain (e.g., see Chapter $6$, \citenptext{grimmett2001probability})  representing the probabilities of one round in the consensus process~\cite{cooper2016linear}. If $G$ is strongly connected, this Markov chain is irreducible and finite, so there exists a unique stationary distribution $\mu$ of $H$, that is, there is a row vector $\mu$ such that $\mu H=\mu$. We call the values $\mu(v)$ the influence of the vertex $v$ in the consensus protocol. In this context, previous work~\cite{cooper2016linear}, show that the winning probabilities of each colour can be determined by the initial configuration only and are given by the following proposition.

    \begin{proposition}[\citenp{cooper2016linear}]\label{prop:nicola's-result}
        Consider a consensus process on a strongly connected graph $G$ (further, we assume $G$ is such that consensus is always achieved for all initial configurations), with associated adjacency matrix $H$ and $\mu$ its unique stationary distribution. Assume the initial configuration is given by $\state\in\{c_1,\dots,c_k\}^{V}$. 
        Then, we have that the winning probability of colour $c_i$ is: 
        \[\mathbb{P}(\textrm{colour }c_i\textrm{ wins}\mid \STATE_0 = \state) = \sum_{v \in V, \STATE(v)=c_i}\mu(v)\]
    \end{proposition}

    We can see as a corollary, that for (non-bipartite, connected) undirected graphs, the probability of a given colour winning is simply the number of incident edges in nodes of that colour divided by the $2E$, where $E$ is the number of edges in the graph \cite[Corollary 2.2 and Section 2.3]{hassin2001distributed}. Example \ref{exm:other_example} discusses the idea applied to our motivating example.

    \begin{example}\label{exm:other_example}
    Consider the modified version of Example \ref{exm:motivation} with a different initial configuration: agnostic nodes are instead gnostic and coloured, say, orange. In other words, assume we are under the assumptions of the classical voter model with 3 colours. We have $\mu = \frac{1}{8}(2, 2, 3, 1)$. Then, from the we have the probability of red winning being $\frac{3}{8}$, of blue winning being $\frac{2}{8}$ and of orange winning being $\frac{3}{8}$. 
\end{example}

    \subsection{Rumour Spreading Process}

    A rumour spreading process on a graph represents the process of information `travelling' across the edges to eventually reach all nodes on a graph. More formally, and using the notation for the voter model, we would have two colours, one being `red' and the other representing a node being `uninformed'. Many strategies of information transmission were designed, such as push, pull and push-and-pull. In this work, we will concentrate, as mentioned in Section~\ref{subsec:voter}, on the pull protocol. In it, if node $v$ selects $u$, then $v$ becomes informed if $u$ is informed, otherwise $v$ is unchanged (informed or uninformed). The process can be synchronous or asynchronous.
    
    For the push protocol, node $v$ selects $u$ and pushes its state towards $u$. If $v$ is agnostic, $u$ retains its state, otherwise $u$ adopts the state of $v$. Observe that the push protocol cannot be done synchronously for consensus problems as it is not clear how to resolve the possible ambiguity (nodes $v_1$ and $v_2$ with different gnostic states, both select the same node $u$). Asynchronous push-and-pull would be defined as is standard in rumour spreading theory. A random agnostic node is chosen and performs a pull, and a random gnostic node is chosen and performs a push. Observe that the asynchronous push-and-pull would eventually become equivalent to asynchronous push as agnostic vertices disappear. Lastly, observe that, just like the push protocol, push-and-pull cannot be done synchronously.

    We will use standard techniques and proof ideas of the rumour spreading literature to show that the rumour spreading process is fast in general for the pull version. These are done in Corollaries~\ref{corollary:rumour_general} and~\ref{corollary:rumour_general_matrix} (consequences of Lemma~\ref{lemma:rumour_general}) as well as  Proposition~\ref{prop:rumour_random} where the techniques and methods from the literature allow us to obtain an $O(n\log(n))$ bound for general graphs~\footnote{If we assume the transition matrix is arbitrary, an extra $1/p$ factor is included, where $p$ is the smallest entry in the transition matrix.} and a $O(\log(n))$ bound for random graphs (in the synchronous case, the asynchronous case adds an extra factor $n$ multiplying both). For other known bounds from the literature, see Sections~\ref{sec:discussion} and~\ref{sec:related_work}.
    
    \subsection{Voter Model with Agnostic Nodes}\label{subsec:agnostic}
    
    We now introduce the main concept to be explored in this work. The main difference of the process with agnostic nodes is that there is an asymmetry between gnostic and agnostic nodes: an agnostic node can become gnostic but gnostic nodes cannot become agnostic. A more precise definition is given as follows.

    \begin{definition}[Voter Model with Agnostic Nodes]\label{def:agnostic}
        A (pull) voter model with agnostic nodes generalises the notion of (pull) voter model by changing the rule with which nodes update their colour. As before, at each round $t$, each node $v$ chooses a one of its out-neighbours $u$ proportionally to the weight of the edge in $G$. However,
        \begin{enumerate}
            \item If $u$ is agnostic, $v$ does not change its colour. 
            \item If $u$ is gnostic, $v$ copies colour of $u$.
        \end{enumerate}
    \end{definition}
    At the same time that Definition \ref{def:agnostic} can be seen as a generalisation of the voter model, it can also be seen as a generalisation of the rumour spreading process. A node that is `uninformed' behaves equivalently to an `agnostic' node. The difference being, of course, that we consider more than one rumour spreading in the network, and competing with each other at the same time it influences agnostic (or uninformed) nodes. 

    We now go back to Example \ref{exm:motivation} and solve it by simply accounting for all possible states and their probabilities. 
    
    \begin{example}[Example \ref{exm:motivation} continued]\label{exm:motivation2}
        Recall Example \ref{exm:motivation}. Here we solve it `by hand' to motivate the introduction of a martingale property for this model.  Figure \ref{fig:example1solution} shows all possible configurations for $S_1$, the probabilities of reaching them ($a_i$) and the probability that red wins from each of them ($b_i$) calculated by applying Proposition \ref{prop:nicola's-result}. Note that $\mu = \frac{1}{8}(2, 2, 3, 1)$. We have that $\mathbb{P}(\textrm{red wins}\mid \STATE_0 = \state) = \frac{5}{8}$. 
\tikzset{red/.style={draw=red, circle, fill=red!30, inner sep=0.5}}
\tikzset{blue/.style={draw=blue, circle, fill=blue!30, inner sep=0.5}}

\begin{figure}[b]
\centering
\setlength{\tabcolsep}{2pt} 
\renewcommand{\arraystretch}{1.3} 

\begin{tabular}{cccc}
\begin{minipage}[b]{.23\linewidth}\centering
\begin{tikzpicture}
    \node[red] (A) at (-0.52, 0.3) {\tiny $v_1$};
    \node[red] (B) at (-0.52, -0.3) {\tiny $v_2$};
    \node[red] (C) at (0, 0) {\tiny $v_3$};
    \node[red] (D) at (0.6, 0) {\tiny $v_4$};
    \draw (A) -- (B);
    \draw (A) -- (C);
    \draw (B) -- (C);
    \draw (C) -- (D);
\end{tikzpicture}

\footnotesize (a)\; $a_1=\dfrac{1}{6},\,b_1=1$
\end{minipage}
&
\begin{minipage}[b]{.23\linewidth}\centering
\begin{tikzpicture}
    \node[blue] (A) at (-0.52, 0.3) {\tiny $v_1$};
    \node[red] (B) at (-0.52, -0.3) {\tiny $v_2$};
    \node[red] (C) at (0, 0) {\tiny $v_3$};
    \node[red] (D) at (0.6, 0) {\tiny $v_4$};
    \draw (A) -- (B);
    \draw (A) -- (C);
    \draw (B) -- (C);
    \draw (C) -- (D);
\end{tikzpicture}

\footnotesize (b)\; $a_2=\dfrac{1}{6},\,b_2=\dfrac{3}{4}$
\end{minipage}
&
\begin{minipage}[b]{.23\linewidth}\centering
\begin{tikzpicture}
    \node[red] (A) at (-0.52, 0.3) {\tiny $v_1$};
    \node[blue] (B) at (-0.52, -0.3) {\tiny $v_2$};
    \node[red] (C) at (0, 0) {\tiny $v_3$};
    \node[red] (D) at (0.6, 0) {\tiny $v_4$};
    \draw (A) -- (B);
    \draw (A) -- (C);
    \draw (B) -- (C);
    \draw (C) -- (D);
\end{tikzpicture}

\footnotesize (c)\; $a_3=\dfrac{1}{6},\,b_3=\dfrac{3}{4}$
\end{minipage}
&
\begin{minipage}[b]{.23\linewidth}\centering
\begin{tikzpicture}
    \node[blue] (A) at (-0.52, 0.3) {\tiny $v_1$};
    \node[blue] (B) at (-0.52, -0.3) {\tiny $v_2$};
    \node[red] (C) at (0, 0) {\tiny $v_3$};
    \node[red] (D) at (0.6, 0) {\tiny $v_4$};
    \draw (A) -- (B);
    \draw (A) -- (C);
    \draw (B) -- (C);
    \draw (C) -- (D);
\end{tikzpicture}

\footnotesize (d)\; $a_4=\dfrac{1}{6},\,b_4=\dfrac{1}{2}$
\end{minipage}
\\[0.6cm]
\begin{minipage}[b]{.23\linewidth}\centering
\begin{tikzpicture}
    \node[red] (A) at (-0.52, 0.3) {\tiny $v_1$};
    \node[red] (B) at (-0.52, -0.3) {\tiny $v_2$};
    \node[blue] (C) at (0, 0) {\tiny $v_3$};
    \node[red] (D) at (0.6, 0) {\tiny $v_4$};
    \draw (A) -- (B);
    \draw (A) -- (C);
    \draw (B) -- (C);
    \draw (C) -- (D);
\end{tikzpicture}

\footnotesize (e)\; $a_5=\dfrac{1}{12},\,b_5=\dfrac{5}{8}$
\end{minipage}
&
\begin{minipage}[b]{.23\linewidth}\centering
\begin{tikzpicture}
    \node[blue] (A) at (-0.52, 0.3) {\tiny $v_1$};
    \node[red] (B) at (-0.52, -0.3) {\tiny $v_2$};
    \node[blue] (C) at (0, 0) {\tiny $v_3$};
    \node[red] (D) at (0.6, 0) {\tiny $v_4$};
    \draw (A) -- (B);
    \draw (A) -- (C);
    \draw (B) -- (C);
    \draw (C) -- (D);
\end{tikzpicture}

\footnotesize (f)\; $a_6=\dfrac{1}{12},\,b_6=\dfrac{3}{8}$
\end{minipage}
&
\begin{minipage}[b]{.23\linewidth}\centering
\begin{tikzpicture}
    \node[red] (A) at (-0.52, 0.3) {\tiny $v_1$};
    \node[blue] (B) at (-0.52, -0.3) {\tiny $v_2$};
    \node[blue] (C) at (0, 0) {\tiny $v_3$};
    \node[red] (D) at (0.6, 0) {\tiny $v_4$};
    \draw (A) -- (B);
    \draw (A) -- (C);
    \draw (B) -- (C);
    \draw (C) -- (D);
\end{tikzpicture}

\footnotesize (g)\; $a_7=\dfrac{1}{12},\,b_7=\dfrac{3}{8}$
\end{minipage}
&
\begin{minipage}[b]{.23\linewidth}\centering
\begin{tikzpicture}
    \node[blue] (A) at (-0.52, 0.3) {\tiny $v_1$};
    \node[blue] (B) at (-0.52, -0.3) {\tiny $v_2$};
    \node[blue] (C) at (0, 0) {\tiny $v_3$};
    \node[red] (D) at (0.6, 0) {\tiny $v_4$};
    \draw (A) -- (B);
    \draw (A) -- (C);
    \draw (B) -- (C);
    \draw (C) -- (D);
\end{tikzpicture}

\footnotesize (h)\; $a_8=\dfrac{1}{12},\,b_8=\dfrac{1}{8}$
\end{minipage}
\end{tabular}

\caption{Every configuration $s_{1i}$ which can be reached from $S_0$ in Example~\ref{exm:motivation} after one round, i.e. $a_i:=\mathbb{P}(S_1=s_{1i}|S_0)>0$. The probability of a red consensus in each case is denoted by $b_i$ and can be calculated by applying Proposition~\ref{prop:nicola's-result}. Therefore, the probability of red to win in Example~\ref{exm:motivation} is $\sum_{i=1}^8a_ib_i=\frac{5}{8}$.}
\label{fig:example1solution}
\end{figure}
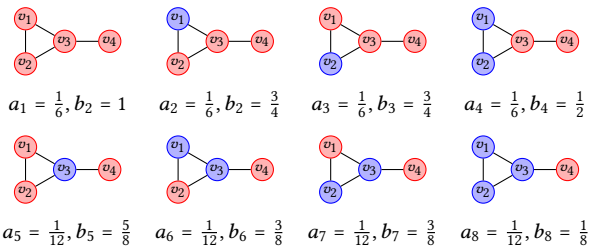

    \end{example}

\begin{remark}\label{rmk:one_round_probabilities}
    Example \ref{exm:motivation2} highlights the following: if we know the probabilities that agnostic nodes become, say, red, once they become gnostic, can we amend formula in Proposition \ref{prop:nicola's-result} so that it is valid for the voter model with agnostic nodes? In Figure \ref{fig:motivational_example}, such probabilities for nodes $v_2$ and $v_4$ are, respectively, $\frac{1}{2}$ and $1$. Multiplying these values with the importance of each node and summing them up gives us, $0\mu(v_1) + \frac{1}{2}\mu(v_2) + 1\mu(v_3) + 1\mu(v_4) = \frac{5}{8}$, which coincides with the probability of red winning. We will show that this is not a coincidence in Theorem \ref{thm:martingale}.
\end{remark}

\section{Probability of Consensus}

In this section, we will provide a general study on the probability of consensus for the (pull) voter model with agnostic nodes.
Our theorems deal with the case with only two colours (red and blue) for gnostic nodes. It is easy to see that the generalisation for more colours is immediate.

\subsection{Exact computation of consensus probability}

Theorem~\ref{thm:probability} provides a method for computing the consensus probability associated with the voter model with agnostic nodes, as a limit involving the left eigenvector of transition matrix $H$, the initial state $S_0$ and quantities corresponding to probabilities that a red vertex selects an agnostic neighbour at round $t$. Theorem~\ref{thm:martingale}, then, assumes a reversibility property of matrix $H$ and produces the same type of martingale (see, e.g., \cite[Chapter 12]{grimmett2001probability}), as is usual in the literature. The original work~\cite{gauy2025vmmrs} presented only the martingale and the immediate consensus probability formula that follows from it.

\begin{theorem}[General consensus probability formula] \label{thm:probability}
    Let $G$ be a graph with vertex set $V$. Suppose we have a (pull) voter model with agnostic states (either synchronous or asynchronous) with associated matrix $H$ on $G$, and $\mu$ such that $\mu H = \mu$. Let $S_t$ be the state of each vertex at time $t$. Define $\mathcal{S}_t^x(v)$ as the probability that vertex $v$ has colour $x$ at round $t$. Moreover, define $P_t(v)$ the probability that vertex $v$ is red at round $t$ and selected an agnostic neighbour at round $t$. If $G$ is such that the process always converges~\footnote{For some $G$, such as bipartite graphs, the process may never converge. For general results on convergence, see \cite{kohan2017team, marzag2021influence, manohara2025viral}.} on $G$, then, if $p$ denotes the consensus probability, the following holds:
    \begin{equation}
        p = \mu S_0 + \mu\sum_{i=0}^\infty P_t
    \end{equation}
\end{theorem}

\begin{proof}
      
      Let $\mathcal{S}_t^x$ be a vector denoting the probability that a vertex at time $t$ gets colour $x$. As every vertex selects a neighbour independently, with probability given by matrix H, to adopt its colour, given knowledge of $\mathcal{S}_t^x$, we can compute the probability that a given vertex, at time $t$, selects a neighbour with colour $x$ by doing $H\mathcal{S}_t^x$. Let $x=r$ to denote the colour red and we can associate $\mathcal{S}_{t+1}^x$ and $\mathcal{S}_t^x$ as follows. For a vertex $v$ to have colour red at time $t+1$, it either chose a red neighbour in round $t$ or it had colour red at round $t$ and selected an agnostic neighbour. Let $P_t(v)$ denote the probability that vertex $v$ was red at round $t$ and selected an agnostic neighbour at round $t$. Sadly, we cannot relate $P_t(v)$ and $\mathcal{S}_t^r(v)$ as independence does not hold. Still, since $P_t$ is a disjoint event from selecting a red neighbour in round $t$, we can write: $S_{t+1}(v) = (HS_t)(v) + P_t(v)$, where we use the definition of $P_t$ and the previous observation that the probability of selecting a red neighbour is equal to the $v$ entry of $H\mathcal{S}_t^r$. Thus, we can write in vector format the expression: $\mathcal{S}_{t+1}^r = H\mathcal{S}_t^r + P_t$.
      
      We can compute the probability of red consensus by computing the limit of $\mathcal{S}_{t}^r$, when $t$ goes to infinity, as a function of $P_t$. For that, we note that we can rewrite the previous formula as: $\mathcal{S}_{t+1}^r - H\mathcal{S}_{t}^r = P_t$. Now, we can sum over $t$ from $0$ to a given value $\tau$ and get: 
      $\sum_{t=0}^{\tau}\mathcal{S}_{t+1}^r - H\mathcal{S}_{t}^r = \sum_{t=0}^{\tau}P_t$. 
      The key idea now is to take the dot product of both sides by the left eigenvector $\mu$  of H (with eigenvalue 1) to obtain the following: 
      $\sum_{t=0}^{\tau}\mu \mathcal{S}_{t+1}^r - \mu H\mathcal{S}_{t}^r = \sum_{t=0}^{\tau}\mu P_t$. 
      Now, observe that $\mu H=\mu$, so the terms on the left side mostly cancel each other and we are left with: $\mu \mathcal{S}_{\tau+1}^r = \mu \mathcal{S}_0^r + \sum_{t=0}^{\tau}\mu P_t$. 
      Now, we assume that the process converges, which implies that the limit of $\mathcal{S}_{t}^r$ is a vector where all values are equal to the probability $p$ that red consensus is reached. Therefore, the limit of $\mu \mathcal{S}_t^r$ when $t$ goes to infinity is exactly $p$ (as the limit is linear). 
      Thus, the previous equation results in $p = \mu \mathcal{S}_0^r + \mu \sum_{t=0}^{\infty} P_t$. 
\end{proof}

As observed in the above proof, computing $P_t$ appears to be hard. We remark that the authors of~\cite{manohara2025viral} erroneously claim that 'since picking is independent across rounds', then $P_t(v)$ can be computed as the product of probability that vertex $v$ is red at round $t$ (given by $\mathcal{S}_t^r(v)$) and the probability that it selects an agnostic neighbour at round $t$ (given by the $v$-entry of $H\mathcal{S}_t^r$). However, these probabilities are not independent and computing $P_t$ does not appear to be simple. Section~\ref{subsubsec:examples} presents a simple counterexample showing that this independence property does not hold. Section~\ref{subsec:complexity} presents, among other theorems, a simple correction to the results of~\cite{manohara2025viral}.

We now state the Martingale property, which assumes a reversibility property of matrix $H$. For completeness, we include the proof of the martingale property, already present in the original paper~\cite{gauy2025vmmrs}. In Section~\ref{subsubsec:examples}, we also present a simple example showing that reversibility is required for the martingale property to hold.

\begin{theorem}[Martingale Consensus Probability Formula]] \label{thm:martingale}
    Let $G$ be a graph with vertex set $V$. Suppose we have a (pull) voter model with agnostic states (either synchronous or asynchronous) with associated matrix $H$ on $G$, and $\mu$ such that $\mu H = \mu$. Suppose, furthermore, that the voter model is a reversible Markov chain, or in other words that, $H(v,w)\mu(v)=H(w,v)\mu(w)$. Let $S_t$ be the state of each vertex at time $t$. Moreover, define $R(v)$ as the event where $v$ is coloured red once it is gnostic. Then the following sequence is a Martingale with respect to $S_t$:
    \begin{equation}
        X_t = \sum_{v\in V}\mu(v)\mathbb{P}(R(v)\mid S_t)
    \end{equation}

    As a result, if $G$ is such that the process always converges on $G$, then
    \begin{equation}\label{eq:consensus}
        \mathbb{P}(\text{red wins}) = X_0.
    \end{equation}
\end{theorem}

\begin{proof}
First note that if $v$ is already red in $S_t$, then $\mathbb{P}(R(v)\mid S_t) = 1$ and if already blue, $\mathbb{P}(R(v)\mid S_t) = 0$. Let $S_t(v)=\{0,1,2\}$ denote the state of each vertex $v$ at time $t$, with $0$ representing color white, $1$ representing red, and $2$ representing blue. We want to show that \begin{align}\label{eq:main}\mathbb{E}(X_{t+1}|S_t)= X_t.\end{align} 
Observe that $X_{t+1}=\sum_{v\in V}\mu(v)\mathbb{P}(R(v)| S_{t+1})$, and we have, by linearity of expectation: 
$$\mathbb{E}(X_{t+1}\mid S_t) = \sum_{v\in V}\mu(v)\mathbb{E}(\mathbb{P}(R(v)\mid S_{t+1})\mid S_t).$$ 

We distinguish two cases: 1) $S_t(v)$ gnostic (that is equal $1$ or $2$) and 2) $S_t(v)$ agnostic (equal $0$). For the second case observe that 
$$\mathbb{P}(R(v)\mid S_t) = \sum_{S_{t+1}}\mathbb{P}(R(v)\mid S_{t+1})\mathbb{P}(S_{t+1}\mid S_t)$$
by law of total probability. The second expression is the definition of the conditional expectation so we have 
$$\mathbb{P}(R(v)\mid S_t) = \mathbb{E}(\mathbb{P}(R(v)\mid S_{t+1})\mid S_t)$$
for the case where we assume $S_t(v)=0$.
We have
$$\mathbb{E}(X_{t+1}\mid S_t) = \sum_{v\in V}\mu(v)\mathbb{E}(\mathbb{P}(R(v)\mid S_{t+1})\mid S_t).$$

Since $X_t = \sum_{v\in V}\mu(v)\mathbb{P}(R(v)\mid S_t)$, the goal is then to show that: 
$$\sum_{v\in V}\mu(v)\mathbb{P}(R(v)\mid S_t) = \sum_{v\in V}\mu(v)\mathbb{E}(\mathbb{P}(R(v)\mid S_{t+1})\mid S_t)$$

For that we split both left and right side in two sums:
$$\sum_{S_t(v) = 0}\mu(v)\mathbb{P}(R(v)\mid S_t) = \sum_{S_t(v)=0}\mu(v)\mathbb{E}(\mathbb{P}(R(v)\mid S_{t+1})\mid S_t)$$

and:
$$\sum_{S_t(v)\in\{1,2\}}\mu(v)\mathbb{P}(R(v)\mid S_t) = \sum_{S_t(v)\in\{1,2\}}\mu(v)\mathbb{E}(\mathbb{P}(R(v)\mid S_{t+1})\mid S_t)$$

that is, we look at the graph at time $t$ and split the vertices between agnostic and gnostic and we will show that the sum is equal in each part. The first case of only agnostic vertices follows from $\mathbb{P}(R(v)\mid S_t) = \mathbb{E}(\mathbb{P}(R(v)\mid S_{t+1})\mid S_t)$ which is why have shown this equality. The second case is done in the rest of the text and requires reversibility.

The equation $\mathbb{E}(X_{t+1}|S_t)= X_t$ then becomes equivalent to showing that:
$$\sum_{S_{t}(v)=1}\mu(v) = \sum_{S_{t}(v)\in\{1,2\}} \mu(v)\mathbb{E}(\mathbb{P}(R(v)\mid S_{t+1})\mid S_t)$$

because we have already shown that the terms on the left and right side where $S_t(v)=0$ are equal.

We now turn to the case where $S_t(v)\neq 0$. Observe that 
$$\mathbb{E}(\mathbb{P}(R(v)\mid S_{t+1})\mid S_t) = \mathbb{P}(S_{t+1}(v)=1\mid S_t),$$ 
as if $S_t(v)$ is gnostic, then $S_{t+1}(v)$ is also gnostic. Now, in terms of $H$, we can write $\mathbb{P}(S_{t+1}(v)=1|S_t)$ equal to the sum $\sum_{S_{t}(w)\in\{0,1\}} H(v,w)$ if $S_t(v)=1$ and equal to $\sum_{S_{t}(w)=1} H(v,w)$ if $S_t(v)=2$.\footnote{These specific formulas assume that we are dealing with a synchronous case. The asynchronous case would have the probability be $\frac{n-1}{n}+\frac{1}{n}\sum_{S_{t}(w)\in\{0,1\}} H(v,w)$ when $S_t(v)=1$ and $\frac{1}{n}\sum_{S_{t}(w)=1} H(v,w)$ when $S_t(v)=2$. The equation that we need to verify stays unchanged after simple manipulation.} Thus, we just need to verify the equality: 
$$\sum_{\scriptscriptstyle S_{t}(v)=1}\mu(v) = \sum_{\scriptscriptstyle S_{t}(v)=1} \mu(v)\sum_{\scriptscriptstyle S_t(w)\in\{0,1\}} H(v,w) +  \sum_{\scriptscriptstyle S_{t}(v)=2} \mu(v)\sum_{\scriptscriptstyle S_{t}(w)=1} H(v,w).$$ 
Using the fact that $\sum_{w} H(v,w) = 1$ we can rewrite the last equality as being equivalent to:

$$\sum_{S_{t}(v)=1} \mu(v)\sum_{S_t(w)=2} H(v,w) = \sum_{S_{t}(v)=2} \mu(v)\sum_{S_t(w)=1} H(v,w)$$

This last equality is equivalent\footnote{Note that the expression has to hold for the case where there is a single blue and red nodes which shows one side. The other side is similarly simple} to $H(v,w)\mu (v) = H(w,v)\mu (w)$ for every pair of neighbours $w,v$. This holds for example when the graphs are undirected and the choice of neighbour is uniform. Thus, \eqref{eq:main} follows. 

Now we assume $G$ is such that every process converges and prove Equation~\ref{eq:consensus}. It is a known technique to use Doob's Stopping Theorem to go from a martingale to the probability of consensus, and that is what we use to show Equation \ref{eq:consensus}. Doob's Stopping Theorem guarantees that $\mathbb{E}(X_\tau) = \mathbb{E}(X_0) = X_0$, where $\tau$ is the time where consensus is reached. Together with the fact that $\mathbb{E}(X_\tau) = \mathbb{P}(\text{red wins}) \mathbb{E}(X_\tau \mid \text{red wins})  + \mathbb{P}(\text{blue wins}) E(X_\tau \mid \text{blue wins})$ (due to the process converging to either a complete red or blue state), and noting that $\mathbb{E}(X_\tau \mid \text{red wins})=1$ and $\mathbb{E}(X_\tau \mid \text{blue wins})=0$, we have finally that $\mathbb{P}(\text{red wins}) = X_0$.
\end{proof}

The equivalence between the consensus probability formula of Theorem~\ref{thm:probability} and~\ref{thm:martingale} can be seen as follows. We have $\mu \sum_{t=0}^{\infty}P_t = \sum_{t=0}^{\infty}\sum_{v\in V}\mu (v)P_t(v)$. 
We can write the term $\mathbb{P}(R(v)\mid S_0)$ as $\sum_{t=0}^{\infty}\mathbb{P}(v\text{ selects red neighbour at round } t \text{ and was agnostic at } t)$, where we omit the conditioning on $S_0$. Let us denote the term $\mathbb{P}(v\text{ selects red neighbour at round } t \text{ and was agnostic at } t)$ by $Q_t(v)$. Our goal is to show that, given reversibility, $\sum_{t=0}^{\infty}\sum_{v\in V}\mu (v)P_t(v) = \sum_{t=0}^{\infty}\sum_{v\in V}\mu (v)Q_t(v)$. We will simply show that $\sum_{v\in V}\mu (v)P_t(v) = \sum_{v\in V}\mu (v)Q_t(v)$ for every time step $t$. Note that $P_t(v)$ can be written as $\sum_{u\in V}H(v,u)\mathbb{P}(v \text{ red at } t \text{ and } u \text{ agnostic at t})$, as vertex $v$ selects a neighbour $u$ with probability $H(v,u)$ independently. The same can be done with $Q_t(v) = \sum_{u\in V}H(v,u)\mathbb{P}(v \text{ agnostic at } t \text{ and } u \text{ red at t})$. Let's call the inner terms of both formulas $P_t(v,u)$ and $Q_t(v,u)$. It is clear that $P_t(v,u) = Q_t(u,v)$. Thus, we want to show that $\sum_{u,v\in V}\mu(v)H(v,u)P_t(v,u) = \sum_{u,v\in V} \mu(v)H(v,u)P_t(u,v)$. By using reversibility, we have that: $\mu(v)H(v,u) = \mu(u)H(u,v)$, and thus the right hand side can be written as: $\sum_{u,v\in V} \mu(v)H(v,u)P_t(v,u) = \sum_{u,v\in V} \mu(u)H(u,v)P_t(u,v)$. Now, it is a simple matter of observing the right hand side is the same as the left hand side with $u,v$ swapped in name. By summing over all $t$ we get the original formula given in the paper.

\subsubsection{Counterexamples for independence property of $P_t$ and reversibility requirement for martingale}
\label{subsubsec:examples}

One may wonder whether the property $H(v,w)\mu(v)=H(w,v)\mu (w)$ of reversibility is required for the martingale to hold, as this is not a requirement in the generalised voter model (Proposition \ref{prop:nicola's-result}). Here, we give a simple (counter-)example of a graph with 3 nodes such that the voter model on it is not a reversible Markov chain. 

\begin{example}[Counterexample for non-reversible chains]
    Consider the initial configuration $S_0 = s_0$ depicted in Figure \ref{fig:counterexample} in a graph with matrix $H$ given by:
    \begin{equation}
        H = \begin{bmatrix}
        \frac{1}{2} & \frac{1}{2} & 0 \\
        0 & \frac{1}{2} &\frac{1}{2} \\
        \frac{1}{2} & 0 & \frac{1}{2}
        \end{bmatrix}
    \end{equation}
    Solving $\mu H = \mu$, we get $\mu = (\frac{1}{3},\frac{1}{3},\frac{1}{3})$. Note that the chain is not reversible because, e.g., $H(v_1, v_2)\mu(v_1) = \frac{1}{6} \neq 0 = H(v_2, v_1)\mu(v_2)$. We can determine $\mathbb{P}(R(v_3)\mid S_0)$ by solving the equation putting together $S_0$ (Figure \ref{fig:counterexample}a) and the four possible states for $S_1$ (Figures \ref{fig:counterexample}b, \ref{fig:counterexample}c, \ref{fig:counterexample}d, \ref{fig:counterexample}e) also using that $\mathbb{P}(R(v_3)\mid S_0 = s_0) = \mathbb{P}(R(v_3)\mid S_1= s_{14})$, i.e, 
    \begin{equation}\label{eq:matingale}
        \mathbb{P}(R(v_3)\mid S_0 = s_0) = \frac{1}{4}\left(1 + 1 + 0 + \mathbb{P}(R(v_3)\mid S_0) \right)
    \end{equation}
    We get that $\mathbb{P}(R(v_3)\mid S_0 = s_0) = \frac{2}{3}$. From that, we get that $X_0 = \frac{5}{9}$ and that, on average, $X_1$ is given by $\mathbb{E}(X_1 \mid S_0) = \frac{1}{4}\left(\frac{1}{3} + \frac{2}{3} + 0 + \frac{5}{9} \right)= \frac{7}{18} \neq X_0$. Another way to see the martingale property does not hold is to evaluate the probability of red winning using a similar technique as in Equation \ref{eq:matingale} to get $\mathbb{P}(\text{red wins} \mid S_0) = \frac{1}{3} \neq X_0$.
\end{example}

\tikzset{red/.style={draw=red, circle, fill=red!30, inner sep=0.5}}
\tikzset{blue/.style={draw=blue, circle, fill=blue!30, inner sep=0.5}}
\tikzset{white/.style={draw=black, circle, fill=white, inner sep=0.5}}

\begin{figure}[t]
\centering
\vspace{-0.5cm}

\begin{minipage}[b]{1\linewidth}\centering
\begin{tikzpicture}
    \tikzset{red/.style={draw=red, circle, fill=red!30, inner sep=1.5}}
    \tikzset{blue/.style={draw=blue, circle, fill=blue!30, inner sep=1.5}}
    \tikzset{white/.style={draw=black, circle, fill=white, inner sep=1.5}}
    \node[red] (A) at (-0.8, 0) { $v_1$};
    \node[blue] (B) at (0, 1) { $v_2$};
    \node[white] (C) at (0.8, 0) {$v_3$};
    \path[->] (A) edge node[above left] {$\frac{1}{2}$} (B);
    \path[->] (B) edge node[above right] {$\frac{1}{2}$} (C);
    \path[->] (C) edge node[below] {$\frac{1}{2}$} (A);
    \path[->] (A) edge[loop left, min distance=3mm, in=160, out=200] node[left] {$\frac{1}{2}$} (A);
    \path[->] (B) edge[loop above, min distance=3mm, in=70, out=110] node[above] {$\frac{1}{2}$} (B);
    \path[->] (C) edge[loop right, min distance=3mm, in=340, out=380] node[right] {$\frac{1}{2}$} (C);
\end{tikzpicture}

\footnotesize (a)\; Initial configuration $S_0 = s_0$. 
We have $X_0 = \frac{5}{9}$, but $\mathbb{P}(\text{red wins} \mid S_0) = \frac{1}{3}$.
\end{minipage}

\vspace{0.5cm}

\setlength{\tabcolsep}{2pt}
\renewcommand{\arraystretch}{1.3}
\begin{tabular}{cccc}
\begin{minipage}[b]{.23\linewidth}\centering
\begin{tikzpicture}
    \tikzset{red/.style={draw=red, circle, fill=red!30, inner sep=0.5}}
    \tikzset{blue/.style={draw=blue, circle, fill=blue!30, inner sep=0.5}}
    \tikzset{white/.style={draw=black, circle, fill=white, inner sep=0.5}}
    \node[blue] (A) at (-0.35, 0) {\tiny $v_1$};
    \node[blue] (B) at (0, 0.5) {\tiny $v_2$};
    \node[red] (C) at (0.35, 0) {\tiny $v_3$};
    \path[->] (A) edge (B);
    \path[->] (B) edge (C);
    \path[->] (C) edge (A);
    \path[->] (A) edge[loop left, min distance=3mm] (A);
    \path[->] (B) edge[loop above, min distance=3mm] (B);
    \path[->] (C) edge[loop right, min distance=3mm] (C);
\end{tikzpicture}

\footnotesize (b)\; Config.\ $s_{11}$ with $X_1 = \frac{1}{3}$
\end{minipage}
&
\begin{minipage}[b]{.23\linewidth}\centering
\begin{tikzpicture}
    \tikzset{red/.style={draw=red, circle, fill=red!30, inner sep=0.5}}
    \tikzset{blue/.style={draw=blue, circle, fill=blue!30, inner sep=0.5}}
    \tikzset{white/.style={draw=black, circle, fill=white, inner sep=0.5}}
    \node[red] (A) at (-0.35, 0) {\tiny $v_1$};
    \node[blue] (B) at (0, 0.5) {\tiny $v_2$};
    \node[red] (C) at (0.35, 0) {\tiny $v_3$};
    \path[->] (A) edge (B);
    \path[->] (B) edge (C);
    \path[->] (C) edge (A);
    \path[->] (A) edge[loop left, min distance=3mm] (A);
    \path[->] (B) edge[loop above, min distance=3mm] (B);
    \path[->] (C) edge[loop right, min distance=3mm] (C);
\end{tikzpicture}

\footnotesize (c)\; Config.\ $s_{12}$ with $X_1 = \frac{2}{3}$
\end{minipage}
&
\begin{minipage}[b]{.23\linewidth}\centering
\begin{tikzpicture}
    \tikzset{red/.style={draw=red, circle, fill=red!30, inner sep=0.5}}
    \tikzset{blue/.style={draw=blue, circle, fill=blue!30, inner sep=0.5}}
    \tikzset{white/.style={draw=black, circle, fill=white, inner sep=0.5}}
    \node[blue] (A) at (-0.35, 0) {\tiny $v_1$};
    \node[blue] (B) at (0, 0.5) {\tiny $v_2$};
    \node[white] (C) at (0.35, 0) {\tiny $v_3$};
    \path[->] (A) edge (B);
    \path[->] (B) edge (C);
    \path[->] (C) edge (A);
    \path[->] (A) edge[loop left, min distance=3mm] (A);
    \path[->] (B) edge[loop above, min distance=3mm] (B);
    \path[->] (C) edge[loop right, min distance=3mm] (C);
\end{tikzpicture}

\footnotesize (d)\; Config.\ $s_{13}$ with $X_1 = 0$
\end{minipage}
&
\begin{minipage}[b]{.23\linewidth}\centering
\begin{tikzpicture}
    \tikzset{red/.style={draw=red, circle, fill=red!30, inner sep=0.5}}
    \tikzset{blue/.style={draw=blue, circle, fill=blue!30, inner sep=0.5}}
    \tikzset{white/.style={draw=black, circle, fill=white, inner sep=0.5}}
    \node[red] (A) at (-0.35, 0) {\tiny $v_1$};
    \node[blue] (B) at (0, 0.5) {\tiny $v_2$};
    \node[white] (C) at (0.35, 0) {\tiny $v_3$};
    \path[->] (A) edge (B);
    \path[->] (B) edge (C);
    \path[->] (C) edge (A);
    \path[->] (A) edge[loop left, min distance=3mm] (A);
    \path[->] (B) edge[loop above, min distance=3mm] (B);
    \path[->] (C) edge[loop right, min distance=3mm] (C);
\end{tikzpicture}

\footnotesize (e)\; Config.\ $s_{14}$ with $X_1 = \frac{5}{9}$
\end{minipage}
\end{tabular}

\caption{Counterexample for the conjecture that the Martingale property (Theorem~\ref{thm:martingale}) is valid for non-reversible chains. Note that edge weights were omitted from panels (b)–(e) for readability.}
\label{fig:counterexample}
\end{figure}

Observe that this counterexample also shows that the independence claimed by the authors of~\cite{manohara2025viral} does not hold and computing $P_t$ cannot be done by using said independence. In fact, observe that the quantity $\mathcal{S}_t^a(v_3) = \frac{1}{2^t}$ as the only way for $v_3$ to remain agnostic is for it to select itself $t$ times and each of those happen independently with probability $1/2$. However, if we compute $\mathcal{S}_2^a(v_3)$ using the formulas provided by~\cite{manohara2025viral}, we obtain: $\mathcal{S}_2^a(v_3) = 1/8 \neq 1/4$. We do not know currently whether there exist (non-trivial) graph families for which computing $P_t$ is tractable (due to e.g., independence holding in that case).

\subsection{Special Cases where Probability of Consensus can be Determined}

Calculating either $P_t$ from Theorem~\ref{thm:probability} or $\mathbb{P}(R(v)\mid S_t)$ from the Martingale Property (Theorem \ref{thm:martingale}), as well as consensus probabilities, appears to be hard in the general case (see Section~\ref{subsec:complexity} for a discussion). There are, however, simple scenarios in which the consensus probability is easy to calculate. We present three of such scenarios below.

The first scenario is rather straightforward and takes into account the initial state, and not the graph itself. It is therefore a result that cannot be guaranteed for a family of graphs, but instead is valid when the initial configuration is such that agnostic nodes only have gnostic neighbours. It is immediate to see, in this case, that only one round is enough for the agnostic nodes to be gnostic, which immediately implies $P_t = 0$ for all $t>0$. It is a simple matter to compute $P_0$.

\begin{proposition}[Solution for Isolated Agnostic Nodes]
    Let $G$ be a graph a $S_0$ an initial configuration. We say an agnostic node $u \in I$ is isolated if all nodes in its out-neighbourhood are gnostic at the initial state. Let $u_r = \sum_{v \,\, s.t. \,\, S_0(v) = \text{red}} H(u,v)$ be the sum of  edge weights from $u \in I$ to red nodes. Then, $P_0(u) = u_r$ and $P_t = 0$ for all $t>0$. As a consequence the consensus probality is:
    \begin{equation}\label{eq:1}
        p = \mu S_0 + \mu P_0
    \end{equation}
\end{proposition}

The first case above does not include when agnostic nodes have self loops (otherwise their out-neighbourhood would include agnostic nodes). 
The next special case is for complete graphs $G$, regardless of initial state and how many agnostic nodes. It does assume that the transition matrix is such that a vertex selects a neighbour uniformly.

\begin{proposition}[Solution for Complete Graph]
\label{corollary:complete_graph}
Let $G$ be a complete graph and consider the asynchronous (pull) voter model with agnostic states, with initial state $S_0$. Let $\gamma$ denote the proportion of red nodes among the gnostic ones. Then:
\begin{equation}\label{eq:1}
    \mathbb{P}(R(u)\mid S_0) = \gamma \text{ for all } u\in I
\end{equation}
As a consequence, 
\begin{equation}
    \mathbb{P}(\text{red wins}\mid S_0) = \gamma 
\end{equation}
\end{proposition}

\begin{proof}
    First, note that the complete graph is reversible, and thus we can apply Theorem~\ref{thm:martingale}. We want to compute $\mathbb{P}(R(u)\mid S_0)$ for every vertex $u\in I$. Let $T(u)$ denote the first time $u$ gets a colour. From the law of total probability, we can write 
    $$\mathbb{P}(R(u)\mid S_0) = \sum_{i=1}^{\infty} \mathbb{P}(T(u)=i)\mathbb{P}(R(u)\mid S_0, T(u)=i).$$ 
    It is easy to see that for nodes on $I$, irrespective of $S_0$, we have that $\mathbb{P}(R(u)\mid S_0, T(u)=1)$ equals the proportion of red nodes among the gnostic ones (which is $\gamma$). Thus, we would like to show that $\mathbb{P}(R(u)\mid S_0, T(u)=i)$ equals to  $\gamma$ and we will have concluded the result. For that, we will use induction.
    We need a strong induction hypothesis. Consider an alternative initial configuration $\tilde{S}_0$ with gnostic set that contains (or is equal to) the gnostic set in the original configuration ${S}_0$. We will show by induction that $\mathbb{P}(R(u)\mid \tilde{S}_0, T(u)=i)$ equals the proportion of red nodes in $\tilde{A}_0$. The base case, $i=1$ is easy as observed before. Suppose that for every $\tilde{S}_0$ we have the desired result for a certain $i$. We want to show it for $i+1$. We can now write the above by conditioning on all possible values $\tilde{S}_1$ can take: 
    \begin{multline}
        \mathbb{P}(R(u)\mid \tilde{S}_0, T(u)=i+1) = \sum_{\tilde{S}_1} \mathbb{P}(\tilde{S}_1 \mid \tilde{S}_0, T(u)=i+1)\mathbb{P}(R(u) \mid\tilde{S}_1, T(u) = i+1).
    \end{multline}
    For convenience, we will omit the conditionals on $\tilde{S}_0$ and $T(u)$ from the term $\mathbb{P}(\tilde{S}_1 \mid \tilde{S}_0, T(u)=i+1)$ writing simply $\mathbb{P}(\tilde{S}_1)$. We can write $\mathbb{P}(R(u)|\tilde{S}_1, T(u) = i+1) = \mathbb{P}(R(u)\mid \bar{S}_0=\tilde{S}_1, \tilde{T}(u) = i)$, where we use $\bar{S}_0$ to denote a Markov chain at time $0$ that starts at state $\tilde{S}_1$. By the induction hypothesis, we have that $\mathbb{P}(R(u)|\tilde{S}_1, T(u) = i+1)$ equals the proportion of red nodes in the gnostic set of $\tilde{S}_1$. 

    Let $r$ be the initial number of red nodes in $\tilde{S}_0$ and $b$ be the initial number of blue nodes in $\tilde{S}_0$. We would like to get 
    \begin{equation*}
        \sum_{\tilde{S}_1} \mathbb{P}(\tilde{S}_1)\mathbb{P}(R(u) \mid \tilde{S}_1, T(u) = i+1) = \frac{r}{b+r}.
    \end{equation*}
    Note that the term $\mathbb{P}(R(u) \mid\tilde{S}_1, T(u) = i+1)$ equals $\frac{|\tilde{S}_1=1|}{|\tilde{S}_1\in\{1,2\}|}$ by the induction hypothesis as argued above.  We are left to compute the term $\mathbb{P}(\tilde{S}_1)$ (which is conditioned on $\tilde{S}_0$ and $T(u)=i+1$). Observe that conditioning on $T(u)=i+1>1$ is equivalent to conditioning on vertex $u$ not becoming gnostic in round $1$. Thus, the term $\mathbb{P}(\tilde{S}_1)$ is the transition probability from $\tilde{S}_0$ to $\tilde{S}_1$ divided by the probability that vertex $u$ does not become gnostic in round $1$. First, the probability that vertex $u$ does not become gnostic in round $1$ is $1-\frac{b+r}{n(n-1)}$. Furthermore, note that we can combine multiple states $\tilde{S}_1$ based on the number of red and blue nodes they contain (as the term $\frac{|\tilde{S}_1=1|}{|\tilde{S}_1\in\{1,2\}|}$ depends only on that). It is thus, enough to compute the probability of transitioning into a state with a given number of red and blue nodes as that sums up the probabilities of all the relevant states $\tilde{S}_1$. 
    Using this fact, the states $\tilde{S}_1$ where the number of red and blue vertices does not change are not interesting as they will have exactly the desired target fraction. We are interested in the states where the ratio of red and blue vertices changes. Let $p_{r_1,b_1}$ be the probability that we have exactly $r_1$ red nodes and $b_1$ blue nodes in state $\tilde{S}_1$. 
    As in the asynchronous case, at most one vertex can change state, we only need to care about the terms: 
    $$p_{r+1,b}=\frac{(n-r-b-1)r}{n(n-1)\left(1-\frac{b+r}{n(n-1)}\right)}, p_{r,b+1}=\frac{(n-r-b-1)b}{n(n-1)\left(1-\frac{b+r}{n(n-1)}\right)}$$ 
    for the terms where one agnostic vertex different from $u$ becomes gnostic, and
    $$p_{r+1,b-1}=p_{r-1,b+1}=\frac{br}{n(n-1)\left(1-\frac{b+r}{n(n-1)}\right)}$$
    for the terms where one of the gnostic vertices changes color. 
    The remaining probability mass is on $p_{r,b}$. Thus, we would like to verify that 
    \begin{multline*}
        p_{r+1,b}\frac{r+1}{r+b+1}+p_{r,b+1}\frac{r}{r+b+1}+p_{r+1,b-1}\frac{r+1}{r+b}+p_{r-1,b+1}\frac{r-1}{r+b} 
    \end{multline*}
    is equal to $(1-p_{r,b})\frac{r}{r+b}$. 
    By plugging in the computed values for the $p_{r_1,b_1}$ terms we get:
    $$\frac{(n-r-b-1)r(r+b+1)}{(r+b+1)(n(n-1)-b-r)}+\frac{2br^2}{(r+b)(n(n-1)-b-r)}$$
    which is the same as:
    $$\frac{(n-r-b-1)r(r+b)+2br^2}{(r+b)(n(n-1)-b-r)} = \frac{(n-r-b-1)(r+b)+2br}{(n(n-1)-b-r)}\frac{r}{r+b}$$
    and it can be easily seen that the term multiplying $\frac{r}{r+b}$ equals $p_{r+1,b}+p_{r,b+1}+p_{r+1,b-1}+p_{r-1,b+1}$ which is exactly $(1-p_{r,b})$. Thus, by induction we have $\mathbb{P}(R(u)\mid S_0, T(u)=i)=\frac{r}{r+b}= \gamma$ for every positive integer $i$ (and any starting state $S_0$). As a result, $\mathbb{P}(R(u)\mid S_0) = \sum_{i=1}^{\infty} \mathbb{P}(T(u)=i)\mathbb{P}(R(u)\mid S_0, T(u)=i) = \gamma$.
    Plugging the values in the martingale of Theorem~\ref{thm:martingale}, it is easy to see that the probability of red consensus is also $\gamma$.
\end{proof}

Theorem~\ref{thm:log_agnostic} in Section~\ref{sec:time_bounds} also presents a special case, where we assume that the number of agnostic nodes is at most logarithmic, which can be solved efficiently. The argument of that result relies on a different approach, and we believe it is better to include it in Section~\ref{sec:time_bounds}.

\subsection{Approximating consensus probability with Markov chain Monte Carlo}
\label{subsec:MCMC}

While Theorems~\ref{thm:probability} and~\ref{thm:martingale} gives us a potential means of computing the exact probability of achieving consensus with a certain colour, a generalization of Proposition~\ref{prop:nicola's-result}, i.e., exact computation of $P_t$ or of  $\mathbb{P}(R(v)\mid S_t)$ (from Theorem~\ref{thm:martingale}) appears to be hard. 
With that in mind, we propose evaluating the probabilities of consensus of the voter model with agnostic nodes on general graphs by Algorithm~\ref{alg:MCMC}.

\begin{algorithm}
\caption{Estimating the Probability of Red Consensus}
\label{alg:MCMC}
\begin{algorithmic}[1]
    \State \textbf{Simulate} the process until all nodes are gnostic.
    \State \textbf{Apply} Proposition~\ref{prop:nicola's-result} once all nodes are gnostic to obtain a probability of red consensus $p$.
    \State \textbf{Repeat} the simulation multiple times and compute the average of the probabilities of red consensus $p$ from each run as an unbiased estimate of the true probability.
\end{algorithmic}
\end{algorithm}

By using known results from rumour spreading, Section~\ref{sec:time_bounds} shows that the simulation time required until all nodes are gnostic is at most $O(n^2\log(n))$ for general graphs. Even in the case of general adjacency matrices $H$, an extra $1/p$ factor in the general case is included, where $p$ is the smallest non-zero entry of $H$. As such, all (pull) voter models where the smallest non-zero entry is at most polynomially small, will have the property of polynomial simulation time until all nodes are gnostic. Moreover, the typical case is even faster: only $O(n\log(n))$ time is required for Erd\"os R\'enyi random graphs with high probability.
As usual, 'with high probability' means with probability going to $1$ as the graph size $n$ goes to infinity. The probability space is the product space of the random graph and the consensus process.
In particular, this implies that running the simulation once is not substantially slower than computing the terms $\mu(v)$ (at least on the worst case). In Section~\ref{subsec:complexity}, we provide a theoretical analysis of the number of runs required to achieve a given multiplicative error. In Section~\ref{sec:experiments}, we present experiments showing that only a few runs (typically less than a $200$) are required to obtain a good estimate (standard error below $0.01$) of the consensus probability. In particular, not only running the process until all nodes are gnostic is much faster than waiting until consensus is reached, the error of the estimate obtained from multiple runs is also much lower when we end the process at the point where all vertices are gnostic.

\section{Time Bounds and Computational Complexity}
\label{sec:time_bounds}

\subsection{Convergence Time Bounds for Consensus}

In addition to the question of what are the probabilities of achieving consensus with a particular colour, the community has also focused on determining the time it takes for consensus to be achieved on general graphs. Lemma~\ref{lemma:time_bounds} shows that, when agnostic vertices are present, expected consensus times are bounded by the expected time it takes for agnostic vertices to disappear plus whatever bounds one can prove for the classical voter model.

\begin{lemma}[Time bounds for consensus]
\label{lemma:time_bounds}
    Given a graph $G$ and a (pull) voter model on $G$ with agnostic states (either synchronous or asynchronous). Let $T_c$ denote the time it takes for consensus to be achieved and let $T_a$ denote the time it takes for the agnostic vertices to disappear. Let $S_0$ be the starting configuration. Let $f(n)$ be a bound on the expected time of consensus being reached on $G$ in the classical voter model (without any agnostic vertices) and any initial configuration. Then $\mathbb{E}(T_c|S_0)\leq \mathbb{E}(T_a|S_0)+f(n)$.
\end{lemma}

\begin{proof}
    In order for the process to reach consensus, it is necessary that all agnostic vertices become gnostic first. Thus, $T_c = T_a+T$, where $T$ is the time it takes for consensus to be reached once all agnostic vertices became gnostic. Taking expectations on both sides, and using linearity of expectation, we have that $\mathbb{E}(T_c|S_0) = \mathbb{E}(T_a|S_0)+\mathbb{E}(T|S_0)$. By using law of total expectation we have that $\mathbb{E}(T|S_0) = \sum_{s\in\{r,b\}^V} \mathbb{P}(S_{T_a} = s|S_0)\mathbb{E}(T|S_{T_a}=s, S_0)$ where $V$ is the vertex set of the graph, and $\{r,b\}$ represent the colours set (say, red and blue). Now, observe that once all agnostic vertices become gnostic, the process becomes equivalent to a standard consensus process with whatever configuration $S_{T_a}=s$ was reached by that point. That is, conditioned on $s$ we have by our assumption $\mathbb{E}(T|S_{T_a}=s)\leq f(n)$. Thus $\mathbb{E}(T|S_0) \leq \sum_{s\in\{r,b\}^V} \mathbb{P}(S_{T_a}=s|S_0)f(n) = f(n)$. Adding $\mathbb{E}(T_a|S_0)$ and using linearity of expectation, we get $\mathbb{E}(T_c|S_0) \leq \mathbb{E}(T_a|S_0)+f(n)$.
\end{proof}

Note that Lemma~\ref{lemma:time_bounds} is only useful for graphs $G$ in which consensus is guaranteed. Otherwise, $f(n)$ may be infinite, in which case, although true, the result does not give us any useful information.  As a result of Lemma~\ref{lemma:time_bounds}, all bounds which hold for the standard voter model will also hold for the case where agnostic vertices are present with an extra time for the agnostic vertices to disappear. In practice, the time for agnostic vertices to become gnostic is typically much shorter than the time it takes to reach consensus as a simple consequence of rumour spreading results. Indeed, it is easy to see that the process of agnostic vertices disappearing is analogous to that of standard rumour spreading. However, the typical results in the literature for the rumour spreading process deal only with the push model, or sometimes the push-pull model. We, on the other hand, want results for the pull model. This is only a minor inconvenience though, as the proof ideas from previous works can also be used for the pull model. Lemma~\ref{lemma:rumour_general}, alongside Corollary~\ref{corollary:rumour_general} shows that for general graphs, the expected time it takes in the pull model for the agnostic vertices to disappear is $O(n^2\log(n))$.   Observe that a star graph, with a gnostic node not on the center, as well as a path, would take an expected time of order $n^2$ for the agnostic vertices to become gnostic. Thus, Corollary~\ref{corollary:rumour_general} is almost tight. As we are also interested in the case of a general transition matrix $H$, Corollary~\ref{corollary:rumour_general_matrix} shows that given $H$, the expected time it takes in the pull model for the agnostic vertices to disappear is $O(n^2\log(n)/p)$ for $p$ defined as the minimum non-zero entry of $H$. In order to show Corollaries~\ref{corollary:rumour_general} and~\ref{corollary:rumour_general_matrix}, we show the following general Lemma which is used to conclude both results. This lemma is heavily based on known rumour spreading results and is included here for completeness.

\begin{lemma}
    Given a regular graph $G$ (potentially with loops), and a (pull) voter model with agnostic states (either synchronous or asynchronous), where the vertices choose their neighbours (including potentially itself) according to the entries of a given matrix $H$. Let $D_0$ denote the initial gnostic set of vertices and $D_i$ denote the set of vertices with distance exactly $i$ to $D_0$. Let $d$ denote the max distance between a vertex of the graph and the initial gnostic set $D_0$. Define $p_i$ as the minimum probability that a vertex in $D_i$ selects a neighbour from the subgraph $\cup_{j=0}^{i}D_j$. Let $T_a$ denote the number of rounds it takes for the agnostic vertices to disappear from the graph. Then we have that for every function $C(n)$, the probability $\mathbb{P}\left(T_a>\sum_{i=1}^d C(n)g(n)/p_i\right) \leq 2n^{-C(n)+2}$ and also that $\mathbb{E}[T_a] \leq 2g(n)\sum_{i=1}^d 1/p_i + 4$ where $g(n) = \log(n)$ for the synchronous case and $g(n) = n\log(n)$ for the asynchronous case.
    \label{lemma:rumour_general}
\end{lemma}
\begin{proof}
    Let $D_0$ denote the initial set of gnostic vertices. For every $i>0$, let $D_i$ be the set of vertices of distance exactly $i$ to the set $D_0$. Assuming the vertices from $\cup_{j=0}^{i}D_j$ have already become gnostic, given a factor $C(n)$, we bound the probability that it takes longer than $C(n)g(n)/p_{i+1}$ for all vertices in $D_{i+1}$ to become gnostic, where $p_{i+1}$ is defined as the minimum probability that  a vertex in $D_{i+1}$ picks a neighbour from the subgraph $\cup_{j=0}^{i}D_j$ and $g(n) = f(n)/n$ (ie, $g(n) = \log(n)$ in the synchronous case and $n\log(n)$ in the asynchronous case). If we let $T_v$ denote the time it takes for a vertex $v\in D_{i+1}$ to become gnostic, the probability that it takes at least $C(n)g(n)/p_{i+1}$ rounds for all vertices in $D_{i+1}$ to become gnostic can be written as $\mathbb{P}(\exists v\in D_{i+1},T_v\geq C(n)g(n)/p_{i+1})$.
    By union bound, this is at most $\sum_{v\in D_{i+1}}\mathbb{P}(T_v\geq C(n)g(n)/p_{i+1})$. Now, we bound $\mathbb{P}(T_v\geq C(n)g(n)/p_{i+1})$. Note that vertex $v$ has at least one gnostic neighbour as we assume all vertices from $D_{i}$ have turned gnostic already. In the asynchronous case, it also has a probability of $\frac{1}{n}$ of being picked. Thus, the probability that $v$ turns gnostic in a single step is at least $\frac{p_{i+1}}{h(n)}$, where $h(n) = 1$ in the synchronous case and $h(n)=n$ in the asynchronous case. Therefore, the probability that $v$ remains agnostic for $C(n)g(n)/p_{i+1}-1$ steps is at most $(1-\frac{p_{i+1}}{h(n)})^{C(n)g(n)/p_{i+1}-1} \leq 2n^{-C(n)}$, where we used that $(1-\frac{1}{x})^x\leq e$, that $g(n) = h(n)\log(n)$ and $(1-\frac{p_{i+1}}{h(n)})\geq \frac{1}{2}$. Thus, we have shown that, for all $C(n)$, $\mathbb{P}(T_v\geq C(n)g(n)/p_{i+1}) \leq 2n^{-C(n)}$. As a result, if $T_{i+1}$ denotes the time it takes for all vertices in $D_{i+1}$ to become gnostic, we have 
    \begin{equation*}
    \mathbb{P}(T_{i+1}\geq C(n)g(n)/p_{i+1})\leq \sum_{v\in D_{i+1}}\mathbb{P}(T_v\geq C(n)g(n)/p_{i+1})\leq 2n^{-C(n)+1}.
    \end{equation*}

    Now, we observe that we can bound $T_a$ by the time it takes for the agnostic vertices to disappear in the following auxiliary process: only when all the vertices in $D_{j}$ for $j\leq i$ have turned gnostic, can a vertex in $D_{i+1}$ turn gnostic. Let $T_{i}$ be the time it takes for all vertices in $D_{i}$ to turn gnostic after all vertices in $D_{i-1}$ have turned gnostic in this auxiliary process. We have that $T_a\leq \sum_{i=1}^d T_{i}$, where $d$ is the max distance between a vertex and the initially gnostic set $D_{0}$. The probability that there is a $T_i$ which takes at least $C(n)g(n)/p_i$ rounds is bounded by $2dn^{-C(n)+1} \leq 2n^{-C(n)+2}$ by union bound (where $d$ denotes the maximum distance in the graph for a vertex from $D_0$). 

We can then conclude that:
$$\mathbb{P}\Bigg(\sum_{i=1}^d T_i \geq \sum_{i=1}^dC(n)g(n)/p_i\Bigg) \leq \mathbb{P}(\exists i, T_i \geq C(n)g(n)/p_i)\leq 2n^{-C(n)+2}.$$ As a consequence, we have that:
$$\mathbb{P}(T_a\geq \sum_{i=1}^dC(n)g(n)/p_i\leq 2n^{-C(n)+2},$$ which proves the first part of the lemma. We now bound the expectation of $T_a$.
Let us denote the term $\sum_{i=1}^d 6g(n)/p_i$ by $A$.

If we set $C(n)=2$ and observe that: $\mathbb{E}[T_a] = \sum_{t}t\mathbb{P}(T_a=t) = \sum_{t\leq A}t\mathbb{P}(T_a=t)+\sum_{t>A}t\mathbb{P}(T_a=t)$. We have that 
    $$\sum_{t\leq A}t\mathbb{P}(T_a=t)\leq A\mathbb{P}(T_a\leq A)\leq A$$ and  
    $$\sum_{t>A}t\mathbb{P}(T_a=t)\leq \sum_{t>A}t\mathbb{P}(T_a\geq t).$$
    Observe now that if $t = Kg(n)\sum_{i=1}^d 1/p_i$ for some suitably chosen $K$, then the first part of the lemma implies that $\mathbb{P}(T_a\geq t) \leq 2n^{-K+2}$. Since $t>A$, we must have $K>2$. Thus, we have 
    $$\sum_{t>A}2n^{-c+2}\leq \int_{2}^{\infty} 2n^{-c+2}dc < 4.$$
This concludes the proof of the lemma.
\end{proof}

\begin{corollary}[rumour spreading bounds for general graphs]
    Given the conditions of Lemma~\ref{lemma:rumour_general}, suppose, additionally, that every vertex chooses their neighbours (including potentially itself) we equal probability. Then we have that, for every function $C(n)$, the probability $\mathbb{P}\left(T_a>C(n)ng(n)\right) \leq 2n^{-C(n)/3+2}$ and also that $\mathbb{E}[T_a] \leq 6ng(n) + 4$ where $g(n) = \log(n)$ for the synchronous case and $g(n) = n\log(n)$ for the asynchronous case.
    \label{corollary:rumour_general}
\end{corollary}

\begin{proof}
    Lemma~\ref{lemma:rumour_general} gives us that: $\mathbb{P}\left(T_a>\sum_{i=1}^d C(n)g(n)/p_i\right) \leq 2n^{-C(n)+2}$ and that $\mathbb{E}[T_a] \leq 2g(n)\sum_{i=1}^d 1/p_i + 4$. Thus, we just have to show that: $\sum_{i=1}^d 1/p_i \leq 3n$. Simply observe that: $p_i \geq \frac{1}{d_i}$, since a vertex in $D_i$ has at most $|D_i|+|D_{i+1}|$ neighbours in the subgraph $\cup_{j=i}^{\infty}D_j$ (as its neighbours must be either in $D_i$ or $D_{i+1}$). By definition of $d_i$, and by the fact that the transition probability is uniform, we have that $p_i \geq \frac{1}{d_i}$. Thus, we have that $\sum_{i=1}^d 1/p_i \leq \sum_{i=1}^d d_i \leq \sum_{i=1}^d (1+|D_i|+|D_{i+1}|) \leq 3n$.
\end{proof}

\begin{corollary}
    Given the conditions of lemma~\ref{lemma:rumour_general}, suppose additionally that $\min_{u,v\in V} \frac{1}{H(u,v)} = p$. Then we have that, for every function $C(n)$, the probability $\mathbb{P}\left(T_a>C(n)ng(n)/p\right) \leq 2n^{-C(n)+2}$ and also that $\mathbb{E}[T_a] \leq 2ng(n)/p + 4$ where $g(n) = \log(n)$ for the synchronous case and $g(n) = n\log(n)$ for the asynchronous case.
    \label{corollary:rumour_general_matrix}
\end{corollary}

\begin{proof}
    Simply observe that  if we define $p$ as the minimum probability of selecting a neighbour, we have $p_i\geq p$ for all $i$. Thus, $\sum_{i=1}^{d} 1/p_{i} \leq n/p$ and we can simply use Lemma~\ref{lemma:rumour_general} to conclude the result.
\end{proof}

Constructing a matrix $H$ that requires order $n^2/p$ rounds for agnostic vertices to disappear is rather simple. Consider a path of depth n and suppose only the left endpoint is coloured. Suppose every vertex picks from its left neighbour only with probability p and that the rest of the mass is a loop to itself (the coloured vertex has a probability 1 of picking from itself as it has no left neighbour). The probability the next vertex is coloured is $h(n)/p$ where $h(n)$ is $1$ in the synchronous case and $n$ in the asynchronous case. As this event needs to happen $n$ times we need at least $nh(n)/p$ which is within a factor $\log(n)$ of the bound given by Corollary~\ref{corollary:rumour_general_matrix}.

While the previous lemma and corollaries give bounds for general graphs, it is likely that for a typical graph, the rumour spreading process is significantly faster. Indeed, it has been shown by~\cite{fountoulakis2010reliable, panagiotou2017asynchronous} that in the case of random graphs, the time it takes for agnostic vertices to disappear is of order $n\log(n)$. As their results are for the push model, we state Proposition~\ref{prop:rumour_random} which is for the pull model. In this case, we defer the proof to the appendix as we are just reusing previously known ideas for the push model of the rumour spreading process in the pull case.

\begin{proposition}[rumour spreading bounds for random graphs]
    Let $p>> \log(n)/n$ and $G\sim G(n,p)$ be an Erd\"os R\'enyi random graph. Consider a (pull) voter model on $G$ with agnostic states (either synchronous or asynchronous). Let $T_a$ denote the number of rounds it takes for the agnostic vertices to disappear. Then, with high probability we have that $\mathbb{E}[T_a] = O(f(n)\log(n))$ where $f(n) = 1$ in the synchronous case and $f(n)=n$ in the asynchronous case.
    \label{prop:rumour_random}
\end{proposition}

Lastly, as a consequence of the above results, one can easily see that a single run of our Markov Chain Monte Carlo algorithm is relatively fast, as it is the time to simulate the process until no agnostic vertex remains plus the computation of the influences $\mu(v)$ (which only needs to be done once for the graph). Apart from the number of runs necessary to obtain a good estimate (which is theoretically analyzed in Section~\ref{subsec:complexity} and empirically in Section~\ref{sec:experiments}), the results from this section show that a single run of Algorithm~\ref{alg:MCMC} takes time $O(n^2\log(n))$ in the worst case and $O(n\log(n))$ in the typical case which means MCMC is an efficient way of estimating the probabilities of consensus.

\subsection{Computational complexity of Algorithm~\ref{alg:MCMC} and related problems}
\label{subsec:complexity}

The previous section demonstrated each run of Algorithm~\ref{alg:MCMC} takes polynomially many steps to complete for every matrix $H$ with $\min_{u,v\in V} \frac{1}{H(u,v)}$ being polynomial in the graph size. We now exhibit a simple proof that Algorithm~\ref{alg:MCMC} provides a fully-polynomial time randomized approximation scheme (FPRAS) for the problem of computing the consensus probability, with a multiplicative error rate.

\begin{theorem}
    For every given $\varepsilon$, Algorithm~\ref{alg:MCMC} for estimating the consensus probability $c$ provides an estimate $\bar{c}$ which is, with high probability, in between $(1\pm \varepsilon)c$ by performing at most $k=\frac{2\omega(n)}{\varepsilon^2c^2}$ repetitions and polynomial total runtime, where $\omega(n)$ is an arbitrarily slow growing function of the graph size $n$.
\end{theorem}

\begin{proof}   
   Let us index each run of the Algorithm~\ref{alg:MCMC} by $i$ from $1$ to $k$. Define $X_i$ to be the estimated consensus probability for run $i$. We have: $\mathbb{E}[X_i] = c$, where $c$ denotes the consensus probability. Note that the $X_i$ are independent identically distributed random variables and that $X = \sum_{i=1}^k X_i$ satisfies the condition of Theorem~\ref{thm:azuma} with $c_i=1$ for all $i$ (as $X_i\in [0,1]$). We have that $\mathbb{E}[X] = kc$ and by Theorem~\ref{thm:azuma}, for every given $\varepsilon>0$, the following holds: $\mathbb{P}(\left|X-\mathbb{E}[X]\right|\geq \varepsilon kc)\leq 2\exp\left(-\frac{\varepsilon^2k^2c^2}{2k}\right) \leq 2\exp\left(-\frac{\varepsilon^2kc^2}{2}\right)$. Thus, if we set $k = \frac{2\omega(n)}{\varepsilon^2c^2}$, the probability that the MCMC estimate after $k$ steps has a multiplicative error larger than $1\pm\varepsilon$ is bounded by $e^{-\omega(n)}$, and it is enough to select an arbitrarily slow growing function $\omega(n)$ to conclude the result. By Corollary~\ref{corollary:rumour_general_matrix}, each run takes more than $C(n)g(n)/p$ steps with probability less than $2n^{-C(n)+2}$, by a union bound the probability that the $k$ runs take longer than $kC(n)g(n)/p$ is upper-bounded by $2kn^{-C(n)+2}$, and given $k$ we can select a polynomial $C(n)$ so that this probability is $o(1)$. We just need to verify if $k$ is a polynomial in $n$. It is a simple matter of verifying that $1/c$ is a upper-bounded by a polynomial in $n$. This can be demonstrated by using Theorem~\ref{thm:probability} to conclude that $c$ is at least the minimum smallest entry of the eigenvector $\mu$ of $H$. Now, since $\mu H = \mu$, and the smallest entry of $H$ is $p$, we must have that the smallest entry of $\mu$ is at least $p$ (since $\sum_{v}\mu(v)H(v,w) = \mu(w)$ implies $p = \sum_{v}\mu(v)p \leq \mu(w)$ for every $w$).
   
\end{proof}

Naturally, the reason why we propose an FPRAS approximation for the problem of computing the consensus probability is that we suspect exact computation of said probability is hard. While analyzing complexity of the exact computation of the consensus probability is probably very challenging, previous works have studied related problems and shown their NP-hardness~\citep{zehmakan2024viral, manohara2025viral}. The typical formulation is to study the following problem variant, termed adoption maximization in~\cite{manohara2025viral}.

\begin{definition}
    Consider a voter model with agnostic states on a graph $G$ with transition matrix $H$ and initial state $S$. For a set of vertices $A$, consider the voter model with agnostic states on $G$ with the same matrix $H$, but with initial state $S'$, where $S'(v) = S(v)$ for $v\notin A$ and $S'(v)$ is red for all $v\in A$, and denote by $F_\tau(A)$ the expected number of red vertices after $\tau$ rounds. For a given voter model, a time $\tau$ and a budget $k$, the adoption maximization (AM) problem is to find: 
    $argmax_{A\subset V,|A|\leq k} F_\tau(A)$.
\end{definition}

The original authors~\cite{manohara2025viral} correctly show that the Maximum Coverage problem has a polynomial time reduction to the AM problem, thus implying that a better than $1-1/e$ approximation of the AM problem can only be obtained if $NP\subseteq DTIME(n^{O(\log\log n)})$. However, there is a small mistake in their paper when presenting a deterministic (greedy) approximation of the AM problem which achieves the $1-1/e$ approximation ratio. As argued in Section~\ref{subsubsec:examples}, exact computation of the $P_t$ quantities appears to be hard. Just like observed by the authors of~\cite{zehmakan2024viral}~\footnote{In their work, the gnostic vertices could not change their state}, exact computation of the expected number of red vertices of the voter model after $t$ rounds is likely a hard problem. As such, it is likely not easy to provide a deterministic approximation of the AM problem, as the approach from~\cite{manohara2025viral} relied on computing the expected number of red vertices of the voter model after $t$ rounds. Following~\cite{zehmakan2024viral}, we fix the error in the AM greedy approximation by showing that the greedy algorithm can be approximated by a Monte Carlo (randomized) approach, which produces an $1-1/e-\varepsilon$ approximation ratio for the AM problem. Algorithm~\ref{alg:greedy} presents the proposed greedy algorithm, and follows the standard blueprint for such approximations in the context of the Hill Climbing problem~\cite{chen2013information}.

\begin{algorithm}
\caption{Monte Carlo Greedy Algorithm}
\label{alg:greedy}
\begin{algorithmic}[1]
    \State \textbf{Input:} Graph $G$, initial state $S$, budget $k$, time $\tau$, error factor $\varepsilon > 0$
    \State \textbf{Output:} Selected seed nodes $A$
    
    \State Initialize $A \gets \emptyset$
    \For{$i = 1$ to $k$}
        \State $v \gets \arg\max_{w \in V \setminus A} \text{Est}_\tau(A \cup \{w\})$
        \State $A \gets A \cup \{v\}$
    \EndFor
    \State \textbf{return} $A$

    \Statex

    \Function{Est\_{$\tau$}}{$A$}
        \State $count \gets 0$
        \For{$j = 1$ to $m$}
            \State Simulate $\tau$ rounds of the voter model with initial state obtained from $S$ by making $A$ red
            \State $r \gets$ number of red nodes after $\tau$ rounds
            \State $count \gets count + r$
        \EndFor
        \State \Return $count/m$
    \EndFunction
\end{algorithmic}
\end{algorithm}

The proof that Algorithm~\ref{alg:greedy} approximation is a $1-1/e-\varepsilon$ follows the same method as the one from~\cite{zehmakan2024viral}.

\begin{theorem}
    Algorithm~\ref{alg:greedy} achieves, with high probability, an approximation ratio of $1-1/e-\varepsilon$ in time $O(\frac{9k^2\omega(n)n\tau}{\varepsilon^2p})$, where $p =\min_{u,v} \frac{1}{H(u,v)}$ and $\omega(n)$ is an arbitrary slow growing function. This is polynomial in the graph size if $1/p$ is polynomial.
\end{theorem}

\begin{proof}    
    It is a known classical result (e.g. see Theorem 3.6 of~\cite{chen2013information}) that if function $Est_\tau(A)$ provides an $\eta$-multiplicative error estimate of $F_\tau(A)$ for $\eta = \varepsilon/(3k)$, then Algorithm~\ref{alg:greedy} provides an $1-1/e-\varepsilon$ approximation of the original problem.

    Consider an arbitrary set $A$ and let us study the function $Est_\tau(A)$. Let $X_j$ denote the number of red vertices in the $j$-th iteration of the $Est_\tau$ function divided by $n$. Define $X=\sum_{j=1}^{m} X_j$. It is easy to see that $X$ satisfies the conditions of Theorem~\ref{thm:azuma} with $c_j = 1$, as $X_j \in [0,1]$ are independent (and identical). It is also clear that $\mathbb{E}[X] = mF_\tau(A)$, so by using Theorem~\ref{thm:azuma}, we get that: $\mathbb{P}(|X-\mathbb{E}[X]| \geq \eta\mathbb{E}[X])\leq 2\exp\left(-\frac{\eta^2mF_\tau(A)}{2}\right)$. We simply need to select $m=\frac{2\omega(n)}{\eta^2F_\tau(A)}$ for $\eta = \varepsilon/(3k)$ and some slow growing function $\omega(n)$ and the Monte Carlo greedy approximation is an $\eta$-multiplicative error with probability $1-\exp(-\omega(n))$. Each iteration of $Est_\tau$ needs time at most $O(n\tau)$, and thus $m$ iterations take time $O(mn\tau)$. Since $m=\frac{2\omega(n)}{\eta^2F_\tau(A)}$ and $\eta = \varepsilon/(3k)$, we just have to show that $F_\tau(A)\geq p$ and the time bound given in the statement holds. Note that we have shown in Theorem~\ref{thm:probability} that $\mu \mathcal{S}_{\tau}^r = \mu \mathcal{S}_0^r + \sum_{t=0}^r\mu P_t$ and thus $F_\tau(A) = \mathbb{1}S_\tau^r\geq \mu \mathcal{S}_\tau^r \geq \mu \mathcal{S}_0^r$ and again we use the fact that the minimum entry of $\mu$ is at least the minimum entry of matrix $H$. Thus, $F_\tau(A) \geq p$ and the result is concluded.
    \end{proof}

While the above discussion is interesting and provides related NP-hard problems, we note that it does not have much bearing on the computational complexity of the consensus problem. This is because the consensus probability computation admits a fully-polynomial time randomized approximation, while the AM problem can only be approximated up to $1-1/e$ error in polynomial time. 

We provide one last result, presenting a special case where the consensus problem with agnostic vertices is polynomial. This special case is based on the following lemma:

\begin{lemma}[Linearity Property]
    \label{lemma:linearity_property}
    Consider a voter model with agnostic states on a graph $G$, with transition matrix $H$, and assume the process converges for every starting configuration. Let us denote by $A$ the initial agnostic set, and by $R$ the initial red set. The red consensus probability (for fixed $H$ matrix) depends solely on $R$ and $A$ and may be denoted by $p(R,A)$. Then, the following linearity property holds: 
    $$p(R,A) = \sum_{v\in R} p(\{v\},A),$$
    where $p(v,A)$ denotes the red consensus probability with initial state where $A$ is the agnostic set and the only red vertex is $v$.
\end{lemma}

\begin{proof}
    This can be shown by the following simple argument. When colouring the vertices of the graph, additionally attribute to them their index (that is, we use multiple colours). Red consensus means only the red vertices are left and happens with probability $p(R,A)$ (though there may be multiple vertex indices). We allow the process to continue until all vertex indices collapse into a single one. This is guaranteed as the process converging with any starting configuration with $2$ colours, guarantees it converges for any number of colours. This is a simple induction on the number of colours: simply pick one colour to be red and rename all other colours blue. We either get a red consensus, and thus, the original colour it represents is the winner, or we have a blue consensus, and the state of the original process now has one colour less and induction may be applied. Since all vertex indices eventually become a single one, it is immediate that this process is equivalent to colouring this specific final vertex index with red and all non-agnostic vertices blue. Thus, the event where we have eventual red consensus with initial red set $R$ and agnostic set $A$ is the union of events where we have eventual red consensus with initial red set $\{v\}$ and agnostic set $A$, for all $v\in R$. This concludes the proof.
\end{proof}

The above lemma implies a simple linearity property of the consensus probability. If no agnostic vertices are present, we only need $n$ variables to represent the consensus probability problem, where $n$ is the size of the graph. It is a simple matter to write a set of equations that associates these variables to one another in a system of equations by looking one step ahead of the initial state. This simple idea can be extended to build an algorithm to compute the consensus probability which runs in polynomial time if the number of agnostic vertices present is $O(\log n)$.

\begin{theorem}
    \label{thm:log_agnostic}
    Consider a voter model with agnostic states on a graph $G$, with transition matrix $H$, and assume the process converges for every starting configuration. Let $k$ denote the number of agnostic vertices in the starting configuration. Then, the consensus probability may be computed in time $O(n^32^{3k})$ and additional $O(n^22^{2k})$ memory, by solving a system of linear equations on $O(n2^{k})$ variables.
\end{theorem}

\begin{proof}
    The idea is to construct a system of equations associating the variables to compute the consensus probability. Namely, let $R$ be the initial red set, $A$ be the initial agnostic set and and assume $|A| = k$. We want to compute the consensus probability with the given initial state, denoted by $p(R,A)$. By Lemma~\ref{lemma:linearity_property}, it is enough to compute $p(\{v\}, A)$ for every $v\in R$. Thus, we build an algorithm to compute $p(\{v\}, B)$ for every $B\subseteq A$ and $v\notin B$. Observe that there is a polynomial number of such variables and the number of variables is thus bounded by $n\times 2^{k}$.
    
    We now explain how to compute the variables $p(\{v\}, B)$: with a starting configuration $S_0$ with $v$ as the only red vertex and $B$ as the agnostic set. Consider one step of the algorithm, which reaches state $S_1$, where the red set is $R_1$ and the agnostic set is $B_1\subseteq B$. Since knowledge of $R_1$ and $B_1$ determines $S_1$, we will denote the state $S_1$ by $(R_1,B_1)$. It is simple to build an equation associating $p(\{v\}, B)$ and all possible $(R_1,B_1)$ choices. Namely, $p(\{v\}, B) = \sum_{B_1}\sum_{R_1} \mathbb{P}((R_1,B_1))|S_0)p(R_1,B_1)$. By Lemma~\ref{lemma:linearity_property}, this is the same as: $\sum_{B_1}\sum_{R_1} \mathbb{P}((R_1,B_1)|S_0)p(R_1,B_1) = \sum_{B_1}\sum_{R_1} \mathbb{P}((R_1,B_1)|S_0)\sum_{w\in R_1}p(\{w\},B_1)$. For every fixed $B_1$, we will show that $\sum_{R_1} \mathbb{P}((R_1,B_1)|S_0)\sum_{w\in R_1}p(\{w\},B_1) = \sum_{w\notin B_1} \mathbb{P}(w\text{ red, } B_1 \text{ agnostic set}|S_0)p(\{w\},B_1)$. We can simply swap the summands on the left side and write: $\sum_{R_1} \mathbb{P}((R_1,B_1)|S_0)\sum_{w\in R_1}p(\{w\},B_1) = \sum_{w\notin B_1}p(\{w\},B_1)\sum_{R_1, w\in R_1}\mathbb{P}((R_1,B_1)|S_0)$, since the first sum fixes $R_1$ such that $R_1\cap B_1 = \emptyset$ and then goes over $w\in R_1$, so inverting the summands will go over $w\notin B_1$ and consider all possible choices of $R_1$ which contain that $w$. Now, simply note that $\sum_{R_1, w\in R_1}\mathbb{P}((R_1,B_1)|S_0) = \mathbb{P}(w\text{ red, } B_1 \text{ agnostic set}|S_0)$, since the sum is simply conditioning on all possible choices for the red set in $S_1$ which have $w$ red. As a result, we have that:  $p(\{v\}, B) = \sum_{B_1,w\notin B_1}p(\{w\},B_1)\mathbb{P}(w \text{ red, } B_1 \text{ agnostic set}|S_0)$.
    
    Assuming we are in the synchronous case, every vertex picks an independent neighbour, the factors $\mathbb{P}(w \text{ red, } B_1 \text{ agnostic}|S_0)$ may be computed as follows.  The independence implies that $\mathbb{P}(w \text{ red, } B_1 \text{ agnostic}|S_0) = \mathbb{P}(w \text{ red})\prod_{b\in B_1}\mathbb{P}(b \text{ agnostic})\prod_{b\in B\setminus B_1}\left(1-\mathbb{P}(b \text{ agnostic})\right)$. Observe now that we just need to compute and store the following (for every possible $v$ and $B$): for every $w$, we need to compute the probability that $w$ is red after the first round; for every $b\in B$, we need to compute the probability that $b$ is agnostic after the first round. Since this is just one round these factors may be computed as follows: for every vertex $w\neq v$, the probability that $w$ became red is $H(w,v)$. This can be computed in constant time. For vertex $v$, the probability that it remains red is $\sum_{w\in B}H(v,w)$, as the only way is if it selects an agnostic neighbour. This can be computed in time $O(k)$. Every agnostic vertex $w$ has probability $\sum_{b\in B}H(w,b)$ of remaining agnostic, which may also be computed in time $O(k)$. As there are $O(k)$ factors in the product, the term $\mathbb{P}(w \text{ red, } B_1 \text{ agnostic}|S_0)$ can be computed in $O(k^2)$.  

    Assuming we are in the asynchronous case, an uniform vertex will be selected for changing colour and computing the factors $\mathbb{P}(w \text{ red, } B_1 \text{ agnostic}|S_0)$ is similarly simple. We now have to condition on the selected vertex $z$: $\mathbb{P}(w \text{ red, } B_1 \text{ agnostic}|S_0) = \sum_{z}\frac{1}{n}\mathbb{P}(w \text{ red, } B_1 \text{ agnostic}|S_0,z)$. Note that we assume that $S_0$ has $v$ as the only red vertex and $B$ as the agnostic set so many of the $w$ and $B_1$ choices have to be zero. If we assume $w\neq v$, then it must be the case that $z=w$ which becomes red and $B_1$ can be either $B$ (if $w$ was blue) or $B_1\cup\{w\} = B$ (if $w$ was agnostic). For either case, the desired probability is exactly $\frac{1}{n}H(w,v)$ which can be computed in constant time. If we assume $w=v$, then we must have $B_1 = B$ (if the selected vertex $z$ is gnostic or if it was agnostic and remained agnostic), or we must have $B_1\cup\{z\}=B$ (and the selected vertex $z$ was agnostic and became gnostic). For the first case, the probability is $\sum_{z\notin B}\frac{1}{n} + \sum_{z\in B}\frac{1}{n}\sum_{z_2\in B}H(z,z_2)$, which can be computed in time $O(k^2)$. For the second case, the probability is $\sum_{z\in B}\frac{1}{n}\sum_{z_2\notin B}H(z,z_2)$, which can also be computed in time $O(k^2)$.
    
    Since we have at most $n2^{k}$ coefficients and each can be computed in $O(k^2)$ time, every equation can be built in time $O(k^2n2^{k})$. As there are $O(n2^k)$ equations associating every $p(\{v\}, B)$ for all $B$ and all $v\notin B$, the total memory used is $O(n^22^{2k})$ and runtime is $O(k^2n^22^{2k})$. This is a system of linear equations involving $O(n2^{k})$ many variables which can be solved in $O(n^32^{3k})$ many steps by Gaussian elimination~\footnote{More efficient methods of solving linear equations would give a better algorithm, but we cannot avoid the exponential factor in $k$ with this method.}.
\end{proof}

Observe that the previous result implies a polynomial time algorithm if $k=O(\log n)$. That is, as long as the initial number of agnostic vertices is at most logarithmic in the number of vertices, then it is possible to compute the exact consensus probability in polynomial time.

\section{Experimental Analysis of Algorithm~\ref{alg:MCMC}}
\label{sec:experiments}

Our key approach is to estimate probabilities by using MCMC, performing the simulations only until the point where all nodes are gnostic and using Proposition~\ref{prop:nicola's-result} to get probability values that can later be averaged over many runs to obtain an unbiased estimator for the consensus probability of a certain colour. Results from the rumour spreading theory ensure that each single run is fast. The previous section argued theoretically that not many runs are necessary for Algorithm~\ref{alg:MCMC} to attain good multiplicative errors. In this section, we analyze this problem in practice. Namely, we make a few experiments for certain graph families, in order to have an idea of how many runs are actually necessary in practice. 

Before we describe the experiments in detail, note that an upper bound on the number of experiments required can be trivially obtained by considering the case where we do not stop the algorithm when all nodes become gnostic but instead go all the way until consensus is reached. This will give us a series of zeros and ones (corresponding to the target colour being the consensus one or not), which can then be averaged to get a probability. It is easy to see that each run is distributed as a Bernoulli and therefore their sum is distributed as a Binomial with the number of iterations $it$ and the target consensus probability $p$ as parameters. The average will then have variance equal to $\frac{p(1-p)}{it}$ which implies that the standard error (which is a proxy of the true error in the estimate) will scale as the inverse of the square root of the number of iterations. It is also easy to see that this approach of running the algorithm until consensus is reached gives worse estimates than the one where we stop the algorithm when agnostic nodes disappear.

Our first experiment will show that using our approach is substantially better than the error estimate from the previous paragraph. We run our algorithm for cliques and cycles with $1001$ nodes for a varying number of runs and plot the standard error $\sigma_{\bar{x}}$ of the estimate for it.\footnote{Note that having an odd number of nodes on cycles guarantees the process always converges} We contrast those results with the standard error of the algorithm, where each run goes all the way until consensus is reached (described in the previous paragraph). Observe that Figure~\ref{fig:exp1} shows that the error rates are substantially lower for the case where we use our algorithm versus the case where one runs until consensus is reached~\footnote{Note that, additionally, each individual run takes substantially longer to complete if we wait until consensus is reached}. In fact, even as little as $40$ runs are enough to be below $0.01$ error in both cases. 

Figure~\ref{fig:exp1} also contains an additional line showing the standard error of our algorithm using a connected subgraph with $1001$ nodes and $1925$ edges of the graph representing a social network in Slovakia (Pokec from~\cite{takac2012data}). This subgraph was generated by selecting an initial random vertex $v$ and performing a random walk (depth first search) until $1001$ vertices were selected. Naturally, the line corresponding to the algorithm where we wait for consensus is very close to the result for cliques and cycles as should be expected (since it approximates the standard deviation of a Binomial distribution with number of iterations and target probability $p$ as parameters divided by the number of iterations). Moreover, note that our algorithm is again substantially better than waiting until consensus is reached, just like for cliques and cycles, suggesting that our approach is also very effective for a typical social network. 

\begin{figure}
  \includegraphics{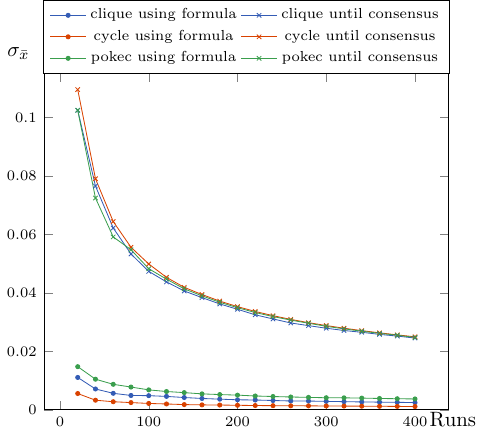}
  \caption{A comparison of the cumulative standard error of the probability of red consensus after all nodes are gnostic and the actual consensus until the simulation finishes. Each simulation ran 400 times on cliques, cycles and a connected subgraph of the Pokec social network with 1001 nodes (5\% red, 5\% blue, 90\% agnostic). For cycles, the initial configuration has all red nodes side by side follow by all blue nodes side by side. For the Pokec subgraph, red and blue nodes were assigned at random, with the rest being agnostic.}
  \label{fig:exp1}
\end{figure}

As an additional experiment, we wonder whether there is a positive effect of increasing the graph sizes on $\sigma_{\bar{x}}$. It turns out that the answer is yes as Figure~\ref{fig:exp2} shows. There, we run our algorithm for graph sizes varying from $300$ to $3000$, with a fixed number of iterations ($400$).\footnote{Using less than $400$ iterations yields a graph with more variance as expected, but the same type of decay.} We perform the experiment on cliques and cycles again and vary the proportion of gnostic nodes as well.

\begin{figure}
  \includegraphics{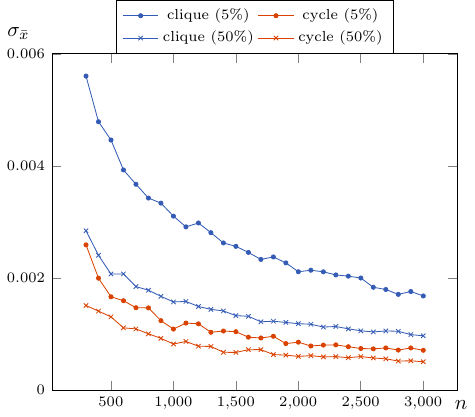}
  \caption{A comparison of the standard error of the probability of red consensus running the simulation 400 times for cliques and cycles of different sizes starting with different proportions of gnostic nodes (5\%, 50\%). Again, the initial configuration for cycles has all red nodes side by side follow by all blue nodes side by side.}
  \label{fig:exp2}
\end{figure}

Lastly, we add that we also performed both experiments on Erd\H{o}s-R\'enyi random graphs with the same values of $n$ and $p=0.05$ and found that the $\sigma_{\bar{x}}$ line for the random graphs essentially stays superposed with the complete graph so we did not find it necessary to include it in the figure to avoid pollution.

For the code repository for these experiments, see \cite{gauy2025vmmrs} or access \url{https://github.com/tmadeira/vmmrs}.

\section{Discussion}
\label{sec:discussion}

For simplicity, in this work, we focused on the pull model as that allows us to analyse both synchronous and asynchronous protocols. However, it is natural to wonder what happens when other strategies are used to transmit information. While synchronous push does not make sense for the voter model, asynchronous push and asynchronous push-and-pull could be used and Proposition~\ref{prop:nicola's-result} shows that the martingale for the voter model holds for asynchronous push. It is likely that both Theorem~\ref{thm:probability} and~\ref{thm:martingale} can be generalized for those strategies as well, requiring only a little extra work, much like the distinction between the asynchronous and synchronous pull protocols.

Much more importantly though, the Markov chain Monte Carlo method of estimating the probabilities of consensus would function in exactly the same way. We could make use of the many known results from the rumour spreading literature to give us guarantees of fast runtime of a single run of Algorithm~\ref{alg:MCMC} on many different types of graphs for those protocols, such as: bounds for the case of general graphs~\cite{feige1990randomized, acan2015push}, random graphs~\cite{fountoulakis2010reliable, panagiotou2017asynchronous, acan2015push}, preferential attachment graphs~\cite{doerr2012asynchronous}, graphs with good conductance~\cite{chierichetti2010rumour} and social networks~\cite{chierichetti2011rumor}. In general, our work implies that one can efficiently estimate probabilities of consensus for a given colour even in the case where agnostic nodes are present as long as a protocol that allows for fast rumour spreading is used.

While we were unable to determine the complexity class of computing the consensus probability for voter models with agnostic nodes, it is important to note that the following problem in the rumour spreading literature is known to be $\#$P-hard: computing the expected number of active nodes given an initial seed set in the independent cascade model~\cite{chen2010scalable}. To the best of our knowledge, a similar result (computing expected number of active nodes after a given number of rounds) for push, pull or push-pull protocols does not exist. Observe that if either of these is found to be $\#$P-hard, then the problem of computing the expected number of red nodes after a given number of rounds in the voter model with agnostic nodes using that protocol would also be $\#$P-hard, as this problem contains the rumour spreading version (just assume the initial state consists only of red and agnostic vertices). The same holds for the version where we assume gnostic vertices do not change colour present in~\cite{zehmakan2024viral}. However, note that the complexity class of the expected number of red vertices after a given number of rounds, while related, does not allow us to determine the complexity class of computing the consensus probability (as the consensus probability is actually the limit of that expectation when the number of rounds grows to infinity).

\section{Conclusion and Future Work}

Here, we introduce a variant of the voter model in which nodes can be agnostic, i.e., have no opinion or colour. Once gnostic, nodes cannot return to being agnostic. This can therefore been seen as a merge between two well-studied processes: the classical voter model and the rumour spreading process.  Our approach allows for efficient estimation of the consensus probabilities for many different information transmission protocols (such as, synchronous pull, asynchronous pull, asynchronous push and asynchronous push-and-pull). We also provide a martingale akin to the one from the classical voter model, and use it to compute exact probabilities of consensus for complete graphs, and initial configurations in general graphs but where there are no edges between agnostic nodes. 

In future work, we consider attempting exact computation of the consensus probabilities for other graph families, like $d$-regular graphs. In general, determining the complexity class of exact computation of consensus probabilities is an interesting and challenging open problem. 
Lastly, it may be of interest to study continuous consensus protocols~\cite{mizrahi2008continuous} in the presence of initially agnostic processes.

\section{Related Work}
\label{sec:related_work}

The notion of reversibility in the context of the voter model has previously been used by Hassin \& Peleg~\cite{hassin2001distributed} to provide winning probabilities on a class of dynamic networks called `stabilising dynamic graphs'. 
In these networks, until a given round, edges may disconnect and nodes attempting to copy a disconnected neighbour keep their own colour, similarly to gnostic nodes choosing agnostic ones in our protocol. 
The authors showed that their previous results hold for networks with reversible Markov chains, i.e. the total influence of the nodes of a given colour remains a martingale in the new process.

Several works studied agents that can be `undecided' as an intermediate state, with nodes transitioning to this state when they select a differently coloured neighbour (\cite{angluin,perron, clementi_et_al:LIPIcs.MFCS.2018.28, petra1}). For a complete graph with binary opinions, the synchronous variant of this protocol has been shown to converge to the most common (plurality) colour in $O(\log n)$ rounds with high probability, assuming there is an initial difference of $\Omega(\sqrt{n \log n})$ in the numbers of agents with each colour~\cite{clementi2018tight}. Similar results have been obtained for the asynchronous protocol in the context of chemical reaction networks~\cite{condon2020approximate}. Additionally, for the consensus problem with $k>2$ opinions, Becchetti et al. defined a `monochromatic distance' function which measures the distance between any colour configuration and consensus, and used this to bound the convergence time of the synchronous process by $O(k \log n)$~\cite{becchetti2015plurality}. Two other famous dynamics are the \textsc{2-Choices} and \textsc{3-Majority}.
The
 \textsc{2-Choices} dynamics, whre an agent samples two agents and only adopts their opinion if they coincide, famously studied in \cite{poweroftwochoices2014} and later on expander graphs \cite{cooperFastConsensusVoting2015,cooperFastPluralityConsensus2017}.
Another closely related dynamics is the
\textsc{3-Majority} dynamics \cite{stabilizingconsensus2016,becchettiSimpleDynamicsPlurality2017} dynamics. Here three agents are sampled (sometimes including the agent itself) the majority opinion is adopted. Ties are broken arbitrarily.
A general analysis of the convergence time of \textsc{2-Choices} and \textsc{3-Majority} was first provided in \cite{ghaffariNearlyTightAnalysis2Choice2018} and refined in \cite{2choices,shimizu3Majority2ChoicesMany2025}.

On the other hand, Demers et al. proposed rumour-spreading protocols to aid the maintenance of distributed databases; these include push, pull, and push-pull transmissions~\cite{demers1987epidemic}. For the synchronous push model, it is known that the number of rounds required to broadcast the rumour to all nodes is at most $O(n \log n)$, which is tight for the star graph~\cite{feige1990randomized}. This process has also been analysed for several other topologies, including complete graphs, hypercubes, bounded-degree graphs, and random graphs~\cite{feige1990randomized}. It has also been studied in the case of random walks spreading the rumour (e.g., \cite{rumourwalk,dimitriou2006infection}).

Our model also has similarities to the biased voter model proposed in \cite{george}, where one colour (corresponding to the agnostic state) has a bias of $0$. However, their results do not apply in our setting since they assume that one colour has a strictly higher preference than all other colours. See also \cite{Lanchier_Neuhauser_2007} for the biased voter model with $2$ opinions in the continuous-time model. Our work is also related to~\cite{zehmakan2024viral}. In it, the authors study the evolution of a process with agnostic nodes, where the key difference is that gnostic nodes can never change colour. Additionally, the work by \cite{manohara2025viral} mentioned previously studies agnostic nodes in the same voter model setting we use. Both of these works focus on the adoption maximization problem and demonstrate that an approximation factor around $1-1/e$ may be achieved, but doing better requires $NP\subseteq DTIME(n^{O(log log n)})$.

Previous works on opinion diffusion have also studied related concepts to agnostic nodes, such as stubbornness. Those are mostly in the context of the majority model~\cite{auletta2017information,out2021majority} and the related Friedkin-Johnson model~\cite{xu2022effects, shirzadi2024stubborn}. The main difference between these models and ours is that the process dynamics is deterministic in the majority and Friedkin-Johnson models, whereas in the voter model, the process dynamics are randomised. 
For the deterministic model, where each agent sets its opinion to the majroity among all of its neighbors see \cite{kaaser2015votingtimedeterministicmajority}.
Lastly, we point the reader to a nice survey on opinion dynamics covering many of the models studied in the literature~\citep{shirzadi2025opinion}.

\subsection*{Code Repository}
For the code repository for these experiments, see \cite{gauy2025vmmrs} or access \url{https://github.com/tmadeira/vmmrs}. 

\subsection*{Acknowledgements}
This work was supported by FAPESP grant number 2022/16374-6 (MMG) and EPSRC grant EP/W005573/1 (FMT).

\newpage

\appendix

\section{Additional Experiments}
\begin{figure}[b]
  \includegraphics{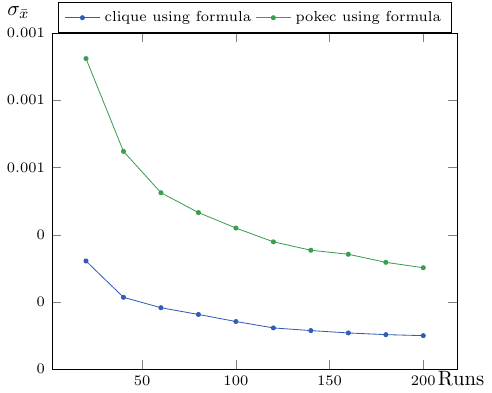}
  \caption{A comparison of the cumulative standard error of the probability of red consensus after all nodes are gnostic. Each simulation ran 200 times on cliques and the full Pokec social network, both with 1,632,803 nodes (5\% red, 5\% blue, 90\% agnostic). For the Pokec graph, red and blue nodes were assigned at random, with the rest being agnostic.}
  \label{fig:exp3}
\end{figure}

As the experiments in Section \ref{sec:experiments} were only done for smaller graphs, here we include additional experiments to show our algorithm at work on larger, real-world graphs. 

The first  additional experiment (Figure~\ref{fig:exp3}) shows the cumulative standard error of the probability of red consensus at the point in which all nodes become gnostic. The graphs compared are the full Pokec social network and a clique of the same size. We perform $200$ repetitions of Algorithm~\ref{alg:MCMC} and plot the standard error curve over the number of runs. Observe that, for large graph sizes the initial error is much smaller than for smaller graph sizes. This is already suggested by Figure~\ref{fig:exp2} but is further confirmed by the experiment on large graphs.

\begin{figure}[t]
    \includegraphics{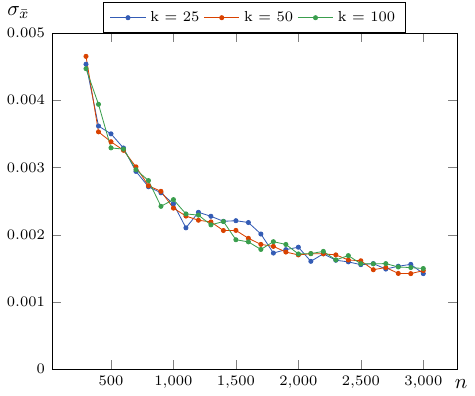}
    \caption{A comparison of the standard error of the probability of red consensus running the simulation 400 times for stochastic block models of different sizes with 10\% of the nodes gnostic (half red, half blue).}
    \label{fig:exp4}
\end{figure}


We also do a second additional experiment (shown in Figure~\ref{fig:exp4}) where the graphs considered are stochastic block models with different community sizes. These are generalizations of Erd\H{o}s-Rényi random graphs where the graph is split into communities and there are different connection probabilities inside each community and between communities. For simplicity, we set the connection probability inside a community to be $1$ and the probability between communities to be $0.05$ and vary the community size $k$. Observe that the standard error $\sigma_{\bar{x}}$ is quite small after $400$ runs, that $\sigma_{\bar{x}}$ decreases as the graph size increases and the community size does not make a significant difference on $\sigma_{\bar{x}}$.

\section{Proof of Proposition~\ref{prop:rumour_random}}

The literature on the rumour spreading process typically dealt with the push model. As we are dealing with a pull model we cannot directly refer to those results. However, the original ideas still work and the proof below is a mere adaptations of the known approaches for the push model. We deal with the case for Erd\H{o}s-R\'enyi random graphs in the push model. We use the same proof ideas from~\cite{panagiotou2017asynchronous} and~\cite{fountoulakis2010reliable} to show Proposition~\ref{prop:rumour_random} which is the pull model for random graphs.

\begin{proposition}[Rumour spreading bounds for random graphs - pull version]
    Let $p>> \log(n)/n$ and $G\sim G(n,p)$ be an Erd\H{o}s-R\'enyi random graph. Consider a (pull) voter model on $G$ with agnostic states (either synchronous or asynchronous). Let $T_a$ denote the number of rounds it takes for the agnostic vertices to disappear. Then, with high probability we have that $\mathbb{E}[T_a] = O(f(n))$ where $f(n) = \log(n)$ in the synchronous case and $f(n)=n\log(n)$ in the asynchronous case.
\end{proposition}

\begin{proof}
   \textbf{Asynchronous case:} For the asynchronous case, we will show that $\mathbb{E}(T_a) \geq \left(2+o(1)\right)n\log(n)$. The idea is simple and borrows from~\cite{panagiotou2017asynchronous}. The main difference in the proof is that we do it for the pull model, instead of the push model. Assuming that $i$ vertices are gnostic, let $t_i$ denote the number of rounds it takes for the vertex of number $i+1$ to become gnostic. We can assume, wlog, that we start with a single gnostic vertex. Then $T_a = t_1+\ldots+t_{n-1}$. By linearity of expectation, we have that $\mathbb{E}(T_a) = \sum_{i=1}^{n-1}\mathbb{E}(t_i)$. Thus, we just need to bound $\mathbb{E}(t_i)$ for each $i$. The crucial step is to notice that conditioned on the current gnostic set $S$ of size $i$, we have that $t_i$ is geometrically distributed with probability of a newly gnostic vertex at any given round given by $\frac{1}{n}\sum_{v\in [n]\setminus S} \frac{d_{S}(v)}{d(v)}$, where $d_{S}(v)$ denotes the degree of vertex v in the set $S$ and $d(v)$ is the degree of vertex $v$. We can assume that $d(v) = (1+o(1))np$ for all $v$ as that happens with high probability. Moreover 
   \begin{equation*}
   \sum_{v\in [n]\setminus S} d_{S}(v) = e(S,[n]\setminus S) = \left(1\pm \sqrt{\frac{8}{\alpha(n)}}\right)(n-|S|)|S|p
   \end{equation*}
   by Lemma~\ref{lemma:kosta} with high probability. The expected time of a geometrically distributed random variable with probability $q$ is given by $1/q$. Thus, by combining all of the above, for all graphs $G$ satisfying the conditions of Lemma~\ref{lemma:kosta}, we have by law of total expectation that $\mathbb{E}(t_i) = \sum_{S}\mathbb{P}(S)\mathbb{E}(t_i|S)$ and $\mathbb{E}(t_i|S) = \frac{n^2p}{\left(1\pm \sqrt{\frac{8}{\alpha(n)}}\right)(n-i)ip}$ (where we use that $|S| = i$). Thus $\mathbb{E}(t_i) = (1+o(1))\frac{n^2}{i(n-i)} = (1+o(1))\left(\frac{n}{i}+\frac{n}{n-i}\right)$. By summing over all $i$ we get 
   $$\mathbb{E}(T_a) = \sum_{i=1}^{n-1}(2+o(1))\left(\frac{n}{i}\right) \leq (2+o(1))n\log(n).$$
    
    \textbf{Synchronous case:} We follow similar steps to the proof of the rumour spreading process with a push protocol for random graphs presented in~\cite{fountoulakis2010reliable}. Let $S_t$ denote the gnostic set at round $t$. We split the process in 2 stages: given a small $\varepsilon>0$, the first stage, which takes time $T_1$, is for the range where $|S_t|\leq \varepsilon n$; the second stage, which takes time $T_2$, is for the range $\varepsilon n\leq |S_t|$.

    For the first stage, we split it in two parts: the time $T_1'$ it takes to activate the first $\log^{1/2}(n)$ vertices and the time ($T_1-T_1'$) it takes to activate $\varepsilon n$ vertices after the first $\log^{1/2}(n)$ were activated. To compute $T_1'$, observe that there are at least $(1+o(1))np$ agnostic vertices with at least one gnostic neighbour (as a single gnostic vertex has $(1+o(1))np$ neighbours and as $np>>\log^{1/2}(n)$ so at least $(1+o(1))np$ of those are agnostic). For every round, the probability that no new agnostic vertex becomes gnostic is at most $\left(1-\frac{1}{(1+o(1))np}\right)^{(1+o(1)np)}\leq ce^{-1}$ (where $c$ is a constant very close to $1$ for large enough $n$, for simplicity set $n$ large enough so that $c<2$) as there are at least $(1+o(1))np$ agnostic vertices with at least one gnostic neighbour and each of them has degree $(1+o(1))np$. As a result, the probability that $T_1' > C\log^{1/2}n$, for some constant $C>1$, is at most $(2/e)^{(C-1)\log^{1/2}(n)} = o(1)$ as there are at least $(C-1)\log^{1/2}(n)$ rounds where no new agnostic vertex became gnostic. 
    
    Thus, $T_1' = o(\log(n))$ with high probability. Let us now compute how long it takes to go from $\log^{1/2}(n)$ gnostic vertices to $\varepsilon n$ gnostic vertices for some small $\varepsilon$ constant. Let $v$ be an agnostic vertex at some round $t$. The probability that $v$ becomes gnostic is $\frac{d_{S_t}(v)}{d(v)}$, where $d_{S_t}(v)$ is the size of the neighbourhood of $v$ in the set $S_t$ and $d(v)$ is the degree of vertex $v$. Let $X_t(v)$ be an indicator random variable for the event that $v$ turns gnostic at round $t$. By the discussion above, $\mathbb{E}[X_t(v)] = \frac{d_{S_t}(v)}{d(v)}$. If $X_t = \sum_{v\notin S_t}X_{t}(v)$ denotes the number of newly acquired vertices at round $t$, then we have that $\mathbb{E}[X_t] = \sum_{v\notin S_t}\mathbb{E}[X_{t}(v)]$, which by the above discussion is equal to $\mathbb{E}[X_t] = \sum_{v\notin S_t}\frac{d_{S_t}(v)}{d(v)}$. By using Lemma~\ref{lemma:kosta}, we can assume that, with high probability, every vertex v has degree $d(v) = (1+o(1))np$ and moreover we can assume that 
    \begin{equation}
        \sum_{v\notin S_t}d_{S_t}(v) = e(S_t, [n]\setminus S_t) = \left(1\pm \sqrt{\frac{8}{\alpha(n)}}\right)(n-|S|)|S|p.
    \end{equation}
    As a result, we have that $\mathbb{E}[X_t] = \left(1\pm \sqrt{\frac{8}{\alpha(n)}}\right)\frac{(n-|S_t|)|S_t|}{n}$. For the first stage, where $|S_t|\leq \varepsilon n$ we cannot use Azuma-Hoeffding's inequality~\footnote{The random variable $X_t$ is the number of agnostic vertices that turn gnostic at round $t$. Note that this is a function of the $X_t(v)$ random variables and that these $X_t(v)$ for a fixed $t$ are independent, as vertices $v\notin S_t$ pick a neighbour independently from one another.} 
    (Theorem~\ref{thm:azuma}), as there are at least $(1-\varepsilon)n$ agnostic vertices and the expectation $\mathbb{E}[X_t] < |S_t|$ so Azuma-Hoeffding won't give meaningful bounds. Talagrand's inequality (Theorem~\ref{thm:talagrand}), however, can give us the bounds we need: changing $X_t(v)$, will change $X_t$ by at most one so the bounded differences condition is satisfied (necessary for both Talagrand and Azuma). For the second condition, observe that $X_t=r$ implies that there are $r$ agnostic vertices which pulled their opinion from gnostic vertices. We can take those agnostic vertices as the set $J$ and the condition will be satisfied with $\psi(r) = \lceil{r}\rceil$. Thus, by Talagrand, if $m$ is a median of $X_t$, we have that $\mathbb{P}(|X_t-m|\geq x)\leq 4\exp\left(\frac{x^2}{\lceil{m+x}\rceil}\right)$. As $\mathbb{E}[X_t]\geq m\mathbb{P}(X_t>m)\geq m/2$, we can rewrite the last inequality in term of $\mathbb{E}[X_t]$:

    $$\mathbb{P}(|X_t-m|\geq x)\leq 4\exp\left(\frac{x^2}{\lceil{2\mathbb{E}[X_t]+x}\rceil}\right]$$

    Moreover, note that by triangle inequality, 
    $$|X_t-\mathbb{E}[X_t]| \leq |X_t-m|+|m-\mathbb{E}[X_t]| = |X_t-m|+O(\sqrt{\mathbb{E}[X_t]}).$$ Since we are assuming that $\log^{1/2}(n)\leq|S_t|\leq \varepsilon n$, we have that $|S_t|\geq \mathbb{E}[X_t]\geq (1-2\varepsilon)|S_t|$ for large enough $n$. Thus, we have that $\mathbb{P}(|X_t-\mathbb{E}[X_t]| \geq \varepsilon |S_t|) \leq \mathbb{P}(|X_t-m|\geq \varepsilon|S_t|/2)$ and by Talagrand $\mathbb{P}(|X_t-\mathbb{E}[X_t]| \geq \varepsilon |S_t|) \leq 4\exp\left(-\frac{\varepsilon^2|S_t|^2}{(2+\varepsilon)|S_t|}\right) \leq 4\exp\left(-\frac{\varepsilon|S_t|}{3}\right)$. 
    
    We have shown that $\mathbb{P}(X_t < (1-3\varepsilon)|S_t|)\leq 4\exp\left(-\frac{\varepsilon|S_t|}{3}\right)$. Now, for computing the time $T_1-T_1'$, we can observe that with probability at most $\sum_{t=T_1'}^{T_1}4\exp\left(-\frac{\varepsilon|S_t|}{3}\right)$, we will have at least one of $X_t < (1-3\varepsilon)|S_t|$ for some $T_1'\leq t\leq T_1$. As $|S_{t+1}|\geq 3/2 |S_t|$ if $X_t\geq (1-3\varepsilon)|S_t|$, the probability that $T_1-T_1'\geq \log_{3/2}(\varepsilon n)$ is at most $\sum_{|S_t|=\log^{1/2}(n)}^{|S_t|=\varepsilon n}4\exp\left(-\frac{\varepsilon|S_t|}{3}\right)$ which is upper bounded by the integral of $4\exp\left(-\frac{\varepsilon|S_t|}{3}{}\right)$ with the same intervals as the sum, which is at most $\frac{12}{\varepsilon}e^{-\varepsilon \log^{1/2}(n)/3} = o(1)$. We conclude the first stage having shown that $T_1\leq \log(n)$ with high probability.

    Now, we move on to the last stage. The idea is again to show that the number of new vertices per round is concentrated around its expectation and said expectation leads to exponential growth. The main advantage is that this time we can just use Azuma-Hoeffding's inequality instead of Talagrand. As before, we will again have $\mathbb{E}[X_t] = (1\pm \sqrt{\frac{8}{\alpha(n)}})\frac{(n-|S_t|)|S_t|}{n}$. This time, however, we can just use Azuma-Hoeffding's inequality as $\mathbb{E}[X_t] = O(n-|S_t|)$. The bounded differences condition again holds with $c_k=1$. As a result, we can apply Azuma's inequality with $c_k=1$ to conclude that $X_t$ is sharply concentrated around $\mathbb{E}[X_t]$. That is, we have 
    $$\mathbb{P}(|X_t-\mathbb{E}[X_t]| \geq \varepsilon \mathbb{E}[X_t])\leq 2\exp\left(-\frac{\varepsilon^2\mathbb{E}[X_t]^2}{2(n-|S_t|)}\right).$$ 
    Replacing the value of $\mathbb{E}[X_t]$, we have 
    $$\mathbb{P}(|X_t-\mathbb{E}[X_t]| \geq \varepsilon \mathbb{E}[X_t])\leq 2\exp\left(-\frac{\varepsilon^2(n-|S_t|)^2|S_t|^2}{4n^2(n-|S_t|)}\right),$$
    where we use that $1\pm \sqrt{\frac{8}{\alpha(n)}}\geq 1/2$. We then get 
    $$\mathbb{P}(|X_t-\mathbb{E}[X_t]| \geq \varepsilon \mathbb{E}[X_t])\leq 2\exp\left(-\frac{\varepsilon^4(n-|S_t|)}{4}\right)$$ 
    where we use that $|S_t|\geq \varepsilon n$.
    Thus,
    $$\mathbb{P}(X_t < (1-\varepsilon)\mathbb{E}[X_t]) \leq 2\exp\left(-\frac{\varepsilon^4(n-|S_t|)}{4}\right).$$
    Observe that $(1-\varepsilon)\mathbb{E}[X_t] \geq (1-2\varepsilon)|S_t|(1-|S_t|/n)$ where we use that $1-2\varepsilon \geq (1-\varepsilon)\left(1\pm \sqrt{\frac{8}{\alpha(n)}}\right)$ for sufficiently large $n$. We have $|S_{t+1}| = |S_t|+X_t$. If $X_t>(1-\varepsilon)\mathbb{E}[X_t]$, we have 
    $$|S_{t+1}|\geq |S_t|+(1-2\varepsilon)|S_t|(1-|S_t|/n) \geq |S_t|(1+(1-2\varepsilon)(1-|S_t|/n)).$$
    We can then write $n-|S_{t+1}|\leq (n-|S_t|)-(1-2\varepsilon)|S_t|(1-|S_t|/n)$, which can be rewritten as $n-|S_{t+1}|\leq (n-|S_t|)(1-(1-2\varepsilon)|S_t|/n)$ and the term multiplying $(n-|S_t|)$ is seen to be at most $(1-\varepsilon(1-2\varepsilon))\leq 1/(1+\varepsilon/2)$. We can then let $T_2'$ be the time it takes to activate $n-\log(n)$ vertices after we have already activated $\varepsilon n$ vertices. From the previous discussion, we either have $T_2' < \log_{1+\varepsilon/2}(n)$ which is $O(\log(n))$ or there is at least one $t$ for which $X_t<(1-\varepsilon)\mathbb{E}[X_t]$, which happens with probability at most $\sum_{|S_t|=\varepsilon n}^{|S_t|=n-\log(n)}2\exp\left(-\frac{\varepsilon^4(n-|S_t|)}{4}\right)$. This is upper bounded by the integral of $2\exp\left(-\frac{\varepsilon^4(n-|S_t|)}{4}\right)$ with the same intervals as the sum, which is at most $\frac{8}{\varepsilon^4}e^{-\varepsilon^4 \log(n)} = o(1)$. Thus, $T_2' = O(\log(n))$ with high probability.
    Now, we have to bound the time it takes to activate the remaining $\log(n)$ vertices. Note that the still agnostic vertices have degree $(1+o(1))np >> \log(n)$, due to the assumption on $p$. This means that every round, each of the still agnostic vertices has probability close to $1$ (easily larger than $1/2$) of turning gnostic (as it has at most $\log(n)$ edges to agnostic vertices). Thus, after $C\log(n)$ rounds, for some $C>0$ the probability that one vertex is still agnostic is at most $\log(n)2^{-C\log(n)}<n^{-C+1} = o(1)$.

    Combining all of the above steps we have that $T_a$ is $O(\log(n))$ with high probability as both stages happen in $O(\log(n))$ with high probability.
    
\end{proof}

\begin{lemma}[Lemma $1$ from~\cite{panagiotou2017asynchronous}]
    Let $p=\alpha(n)\log(n)/n$, where $\alpha(n) = \omega(1)$. Then, with high probability, $G_{n,p}$ is such that for all $S\subseteq [n]$:
    $$e(S,[n]\setminus S) = \left(1\pm \sqrt{\frac{8}{\alpha(n)}}\right)(n-|S|)|S|p$$
    \label{lemma:kosta}
\end{lemma}

\begin{theorem}[Azuma-Hoeffding's inequality]
\label{thm:azuma}
Let $Z_1,\ldots, Z_n$ be independent random variables taking values in the sets $\Lambda_1,\ldots, \Lambda_n$. Let $\Lambda=\Lambda_1\times\ldots\times\Lambda_n$. Let $f:\Lambda\to \mathbb{R}$ be a function and set $X = f(Z_1,\ldots, Z_n)$. Assume there are quantities $c_1,\ldots, c_n$ satisfying the following condition: if $z,z'\in \Lambda$ differ only in the $k$-th coordinate, then $|f(z)-f(z')| \leq c_k$. Then, for every $x\geq 0$ we have that 
$$\mathbb{P}(|X-\mathbb{E}[X]|\geq x) \leq 2\exp\left(-\frac{x^2}{2\sum_{k=1}^n c_k^2}\right).$$
\end{theorem}

\begin{theorem}[Talagrand's Inequality]
\label{thm:talagrand}
Suppose we are in the conditions of Azuma-Hoeffding's inequality. Additionally, assume there is an increasing function $\psi$ satisfying the following: if $z\in\Lambda$ and $r\in \mathbb{R}$ is such that $f(z)\geq r$, then there exists a set $J\subseteq\{1,\ldots, n\}$ with $\sum_{j\in J} c_j^2\leq \psi(r)$, such that for all $y\in \Lambda$ with $y_i=z_i$ for $i\in J$, we have $f(y)\geq r$. 
Then, if $m$ is a median of $X$, we have that, for every $x\geq 0$:
$$\mathbb{P}(|X-m|\geq x) \leq 4 \exp\left(\frac{x^2}{\psi(m+x)}\right).$$
Moreover, when $\psi(r)\leq \lceil{r}\rceil$, we have that $|m-\mathbb{E}[X]|\leq O(\sqrt{\mathbb{E}[X]})$.
\end{theorem}

\bibliography{sn-bibliography}

\end{document}